\documentclass{article}
\usepackage{ijcai23}

\usepackage{times}
\usepackage{soul}
\usepackage{url}
\usepackage[hidelinks]{hyperref}
\usepackage[utf8]{inputenc}
\usepackage[small]{caption}
\usepackage{graphicx}
\usepackage{amsmath}
\usepackage{amsthm}
\usepackage{amssymb}
\usepackage{mathtools}
\usepackage{enumitem}
\usepackage{amsfonts}
\usepackage{multirow}
\usepackage{bm}
\usepackage{bbm}
\usepackage{braket}
\usepackage{booktabs}
\usepackage{algorithm}
\usepackage{algpseudocode}
\usepackage{xcolor}
\usepackage[switch]{lineno}
\usepackage{subcaption}
\usepackage{nicefrac}
\usepackage{tikz,lipsum}
\usepackage[most]{tcolorbox}
\usepackage{natbib}  

\newtheorem{theorem}{Theorem}
\newtheorem{definition}{Definition}
\newtheorem{proposition}{Proposition}
\newtheorem{lemma}{Lemma}

\newcommand{\cB}{\mathcal{B}}

\newcommand{\cI}{\mathcal{I}}

\newcommand{\cM}{\mathcal{M}}

\newcommand{\cR}{\mathcal{R}}

\newcommand{\cX}{\mathcal{X}}

\newcommand{\NN}{\mathbb{N}}

\newcommand{\thresh}{k}

\newcommand{\bound}{B}
\newcommand{\qi}{Q}
\newcommand{\nqi}{N}
\newcommand{\lap}[1]{\mathrm{Lap}\left(#1\right)}

\newcommand{\ds}{D}

\newcommand{\bias}[1]{\bm{\mathcal{B}}\left(#1\right)}

\newcommand{\biasi}[2]{\mathcal{B}\left(#1\right)_{#2}}

\newcommand{\bmech}{\cM_{\text{KA}}^{\beta}}
\newcommand{\prammech}{\cM_{\text{SW}}}
\newcommand{\lapmech}{\cM_{\text{Lap}}}
\newcommand{\tmech}{\cM_{\text{CS}}}
\newcommand{\qiattrs}{n_{\qi}} 
\newcommand{\nqiattrs}{n_{\nqi}} 
\newcommand{\qiuniverse}{\cX_{\qi}}

\newcommand{\acs}{\alpha_{\text{CS}}}
\newcommand{\asw}{\alpha_{\text{SW}}}
\newcommand{\alap}{\alpha_{\text{Lap}}}

\newcommand{\dom}[1]{\operatorname{dom}\left(#1\right)}
\newcommand{\norm}[1]{\left\lVert#1\right\rVert}
\newcommand{\pr}[1]{\operatorname{Pr}\left(#1\right)}

\newcommand{\EE}[2]{\operatorname{\mathbb{E}}_{#1}\left[#2\right]}

\newcommand{\mad}[1]{D_{\mathrm{mean}}\left(#1\right)}

\newcommand{\eps}{\epsilon}

\title{Privacy and Bias Analysis of Disclosure Avoidance Systems}

\author{
Keyu Zhu$^1$\and
Ferdinando Fioretto$^2$\and
Pascal Van Hentenryck$^{1}$\and
Saswat Das$^3$\And 
Christine Task$^4$
\affiliations
$^1$Georgia Institute of Technology\\
$^2$Syracuse University\\
$^3$National Institute of Science Education and Research\\
$^4$Knexus Research Corporation\\
\emails
keyu.zhu@gatech.edu,
ffiorett@syr.edu,
pvh@isye.gatech.edu, 
saswat.das@niser.ac.in,
christine.task@knexusresearch.com
}

\begin{document}

\maketitle

\begin{abstract}
Disclosure avoidance (DA) systems are used to safeguard the confidentiality of data while allowing it to be analyzed and disseminated for analytic purposes. These methods, e.g., cell suppression, swapping, and k-anonymity, are commonly applied and may have significant societal and economic implications. However, a formal analysis of their privacy and bias guarantees has been lacking. This paper presents a framework that addresses this gap: it proposes differentially private versions of these mechanisms and derives their privacy bounds. In addition, the paper compares their performance with traditional differential privacy mechanisms in terms of accuracy and fairness on US Census data release and classification tasks. The results show that, contrary to popular beliefs, traditional differential privacy techniques may be superior in terms of accuracy and fairness to differential private counterparts of widely used DA mechanisms.
\end{abstract}

\section{Introduction}

Disclosure avoidance (DA) systems are methods used to protect confidentiality while still enabling data analyses and dissemination. These techniques are used in various fields, such as economics, public health, social science, and data science, and have a long history in censuses and other data collection efforts. For example, the US Census Bureau has leveraged various traditional DA techniques from the 1930 decennial release on. These include suppressing certain tables based on the number of people or households in a given area and swapping data in records with similar characteristics. 

While traditional confidentiality measures, such as suppression \citep{kelly1992cell}, swapping \citep{dalenius1982data}, and k-anonymity \citep{sweeney2002k} are important for protecting against accidental or intentional disclosure, they lack formal guarantees that quantify the privacy risks that individuals incur upon data releases. This is important as it restricts the ability of participants to assess the impact of these protections on published data. 


In contrast, differential privacy (DP) \citep{Dwork:06} is a \emph{relatively newer} DA that provides a rigorous definition of privacy and allows for quantifiable privacy guarantees. In differential privacy, the privacy of an individual is preserved by adding noise to their data in a controlled way. Such a process ensures that the participation of an individual in a dataset does not significantly affect the results of subsequent queries. 
Marking a significant shift towards more rigorous privacy protections, the US Census has recently adopted differential privacy for the 2020 Census release. However, it is worth noting that many other data agencies and organizations still rely on traditional disclosure avoidance systems to protect the confidentiality of their data. 

While these approaches can be effective at protecting against accidental or intentional disclosures, it is unclear what privacy guarantees they provide when compared to differential privacy. On the other hand, while differential privacy can provide stronger privacy guarantees than traditional disclosure avoidance systems, it may come with a cost in terms of accuracy and fairness \citep{kuppam2019fair,Fioretto:IJCAIa,Fioretto:IJCAI22a}, a topic of considerable debate recently.

Given that these DA are used to release data products that inform decisions with significant societal and economic consequences, it is essential to conduct a rigorous comparison of traditional DA and differential privacy in terms of privacy, bias, and fairness.
However, one of the challenges faced in this comparison is the absence of a standardized framework for evaluating privacy protections. Differential privacy offers a rigorous definition of privacy and enables quantifiable privacy guarantees. On the other hand, traditional disclosure avoidance techniques may not have a distinct set of privacy metrics, making it challenging to directly compare the level of privacy protection they offer.

\begin{figure*}[!t]
\includegraphics[width=\linewidth]{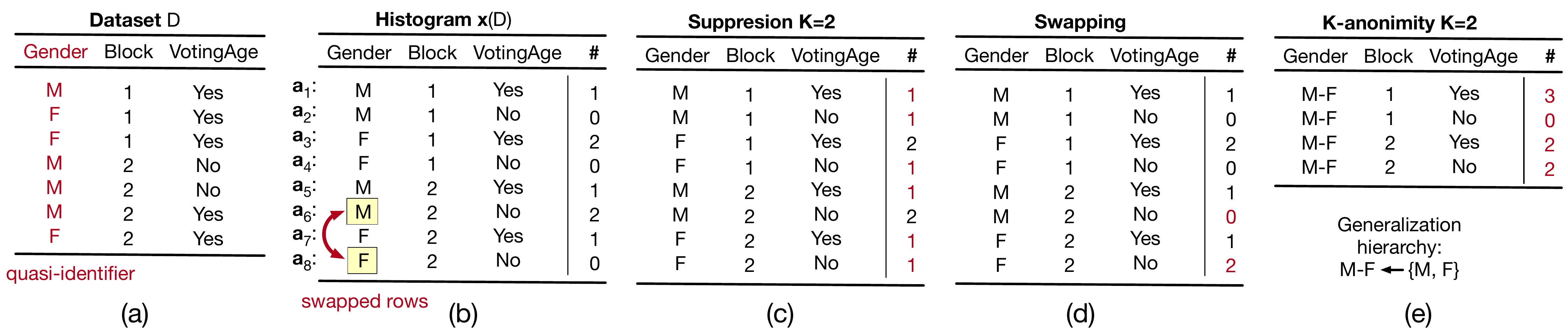}
\caption{\label{fig:scheme} Illustration of the various traditional DA mechanisms.}
\end{figure*}

\paragraph{Contributions} This paper aims at addressing this challenge: it proposes a framework for comparing traditional DA to differential privacy and makes four distinct contributions. (1) It first proposes {\em carefully randomized} versions of three widely adopted {\em traditional} DA: suppression, swapping, and k-anonymity. The resulting randomized mechanisms can then be analyzed rigorously. In particular, the paper  derives $(\epsilon,\delta)$-differential privacy bounds for these new mechanisms and demonstrates that they are close to their traditional counterparts in terms of accuracy. (2) The paper then derives bounds for the bias 
of the new DA mechanisms, allowing for a direct comparison with classical differential privacy techniques for which such bounds exist. 
(3) Next, the paper analyzes the fairness impact induced by the considered DA systems and shows that the fairness violations incurred by the randomized DA algorithms are close to those of their traditional counterparts. (4) Finally, it provides an extensive empirical analysis of the performance of the new DA mechanisms and a comparison with two classical differentially private algorithms on data release and classification tasks. 

{\em From a broader perspective, the paper demonstrates that, contrary to  popular belief, classical differential privacy mechanisms may be superior to traditional disclosure avoidance systems in important data release and learning tasks in terms of accuracy and fairness for the same privacy levels}. As a consequence, the results of this study have the potential to impact the way in which data agencies and organizations approach disclosure avoidance: indeed, it provides the first framework for comparing the relative strengths and limitations of traditional DA and differential privacy.

\section{Problem Setting}

The paper considers datasets of $m$ records with $d$ attributes, $A_1,\dots,A_d$. Each record is a $d$-dimensional tuple of attributes associated with a unique individual from a data universe $\cX\coloneqq \prod_{i=1}^d \dom{A_i}$,
where $\dom{A}$ represents the collection of all the possible values for the attribute $A$.
For convenience, assume that $\cX = \{\bm{a}_1, \ldots, \bm{a}_n\}$, with $n$ being the size of the data universe $\cX$, and consider the \emph{histogram} $\bm{x}(D) \in \NN^n_+$ of dataset $D$, whose $i^{th}$ entry $x_i(D)$ represents the count of the individual records with the combination of attributes $\bm{a}_i$.
When there is no ambiguity, the dataset $D$ is omitted in the expression $\bm{x}(D)$ for simplicity.
Additionally, and without loss of generality, the histogram $\bm{x}$ is assumed to be sorted in some increasing order, i.e., $x_i\leq x_j$, for any $i<j$. Finally, each entry of the histogram $\bm{x}$ is assumed to be bounded by a value $\bound > 0$, i.e., $\bm{x}\in [\bound]^n$. 

Consider, for example, the illustration in Figure \ref{fig:scheme}(a);
The dataset $D$ contains records with three attributes: (geographic)
``Block'', ``Gender'', and ``Voting Age''. The associated histogram is
illustrated in Figure \ref{fig:scheme}(b).  In this instance, the
attribute ``Gender'', when combined with external information like
``Zip code'', can become personally identifying information and
thus is known as a \emph{quasi-identifier} (QI) while the remaining
attributes are referred to as \emph{non-quasi-identifiers}. Throughout
the paper, the sets of quasi-identifiers and non-quasi-identifiers are
denoted by $\qi$ and $\nqi$, respectively. Given a record $\bm{a}$ and
a set $S$ of attributes, $\bm{a}[S]$ is the vector of values for
attributes $S$ in $\bm{a}$.

The goal of the paper is to analyze the privacy, utility, and fairness properties of traditional disclosure avoidance systems (reviewed in the next section) on the task of releasing a privacy-preserving version $\tilde{\bm{x}}(D)$ of the histogram $\bm{x}(D)$. 
The notion of privacy considered in this paper is that of differential privacy, which is reviewed in the next section. The notions of utility and fairness central to the analysis rely on the concept of (statistical) \emph{bias}. For any entry $i\in [n]$, the bias associated with a mechanism $\cM$ is 
\[
    \cB(\cM)_{i}=\EE{}{\cM(\ds)_i}-x_i(\ds)\,,
\]
where the expectation is taken over the randomness of the mechanism, and
\(
    \bias{\cM}=\left[
    \begin{matrix}
        \biasi{\cM}{1}&\dots&\biasi{\cM}{n}
    \end{matrix}
    \right].
\)
Fairness is defined as the maximal difference in biases across the histogram entries.
\begin{definition}[$\alpha$-fairness \citep{ijcai2022p559}]\label{def:a-fair}
    A  mechanism $\cM$ is said to be $\alpha$-fair if the maximum difference among 
    the biases is bounded by $\alpha$, i.e.,
    \begin{equation*}
        \lVert\bias{\cM}\rVert_{\rightleftharpoons} = \max_{i\in[n]}~\biasi{\cM}{i}-\min_{i\in[n]}~\biasi{\cM}{i}\leq \alpha.
    \end{equation*}
\end{definition} 

\section{DA for Private Data Release}
\label{sec:das}

This section provides an overview of the prevalent DA methods utilized by data agencies to safeguard sensitive information within datasets. 
To comply with space limitations, the paper reports the proofs of all theorems in the appendix. 

\paragraph{Differential Privacy.}
Differential privacy (DP) \citep{Dwork:06} is a strong privacy notion used to quantify and bound the privacy loss of an individual participation to a computation. 
 Informally, it  states that the probability of any output does not change much when a record is changed from a dataset, limiting the amount of information that the output reveals about any individual.  
The action of changing a record from a dataset $D$, resulting in a new dataset $D'$, defines the notion of \emph{adjacency}, denoted $D \sim D'$.
\begin{definition}
  \label{dp-def}
  A mechanism $\cM \!:\! \mathcal{D} \!\to\! \mathcal{R}$ with domain $\mathcal{D}$ and range $\mathcal{R}$ is $(\epsilon, \delta)$-differentially private, if, for any two inputs $D \sim D' \!\in\! \mathcal{D}$, and any subset of output responses $R \subseteq \mathcal{R}$:
  \[
      \Pr[\cM(D) \in R ] \leq  e^{\epsilon} 
      \Pr[\cM(D') \in R ] + \delta.
  \]
\end{definition}
\noindent 
Parameter $\epsilon > 0$ describes the \emph{privacy loss} of the algorithm, with values close to $0$ denoting strong privacy, while parameter $\delta \in [0,1)$ captures the probability of failure of the algorithm to satisfy $\epsilon$-DP. 
In particular, the \emph{Laplace mechanism} for histogram data release, defined by
\(
    \cM_{\text{Lap}}(\bm{x}) = \bm{x} + \text{Lap}(\nicefrac{2}{\epsilon}), 
\)
\noindent where $\text{Lap}(\eta)$ is the Laplace distribution centered at 0 and with scaling factor $\eta$, satisfies $(\epsilon, 0)$-DP. 
Additionally, the \emph{discrete Gaussian mechanism} \citep{canonne2020discrete}, defined by
\(
    \mathcal{M}_{\text{Gaus}}(\bm{x}) = \bm{x} + \mathcal{N}_{\mathbb{Z}}(0, \nicefrac{4}{\epsilon^2}), 
\)
where $\mathcal{N}_\mathbb{Z}(0, \sigma)$ is the discrete Gaussian distribution with $0$ mean and standard deviation $\sigma$, satisfies $(\frac{1}{2}\epsilon^2 + \epsilon \sqrt{2 \log(\nicefrac{1}{\delta})}, \delta)$-DP. 

We next discuss three predominant traditional DA systems which, in contrast to differential privacy, do not provide formal bounds on privacy leakage. 

\paragraph{Cell suppression.}
The cell suppression technique \citep{kelly1992cell}, frequently employed by statistical agencies (e.g., \citep{mog}), aims at concealing the low-frequency counts in histograms before data dissemination. 
\begin{definition}
Given a histogram $\bm{x}$ and a threshold value $k$, cell suppression returns a private histogram $\tilde{\bm{x}}$ with entries
\begin{equation}\label{eq:cell_suppression}
\tilde{\bm{x}}_i = \max\left\{\bm{x}_i, \nicefrac{k}{2}\right\}.
\end{equation}
\end{definition}
\noindent
Figure~\ref{fig:scheme}(c) illustrates the application of cell suppression with threshold value $k=2$ to the histogram of Figure~\ref{fig:scheme}(b). The affected row counts are highlighted in red. 
A significant limitation of this approach is that it only protects sensitive attributes with a low number of records while neglecting others.

\paragraph{Swapping.}
Swapping \citep{dalenius1982data} is a mechanism that swaps the values of a set of sensitive attributes (the quasi-identifiers) in a record with those of another record. Informally speaking, the basic steps of the algorithms can be summarized as follows:
\begin{enumerate}[leftmargin=*, parsep=0pt, itemsep=0pt, topsep=2pt]
    \item Select multiple pairs of records in the histogram with probability proportional to their discrepancies;
    \item Swap the values of the quasi-identifiers attributes for each selected pair of records.  
\end{enumerate}

\noindent
Like cell suppression, swapping produces a privacy-preserving histogram $\tilde{\bm{x}}$. However, contrary to cell suppression (and differential privacy mechanisms), swapping requires a piece of additional information: the quasi-identifier attributes of the dataset. Figure~\ref{fig:scheme}(d) illustrates the application of swapping where two rows are swapped, using ``Gender'' as the quasi-identifier attribute (see figure (a)). The affected row counts are highlighted in red. While swapping has been commonly used, for example by the US Census Bureau, to swap similar individuals within close geographies, it is not immune to reconstruction attacks \citep{garfinkel2019understanding}.

\paragraph{$k$-anonymity.\!\!\!\!}
Next, $k$-anonymity protects sensitive data in a dataset by ensuring that each record in the dataset is indistinguishable from at least $k-1$ other records.

\begin{definition}[$k$-Anonymity \citep{sweeney2002k}]\label{def:k-anonymity}
    A dataset satisfies $k$-anonymity, relative to a set of the quasi-identifiers, if and only if when the dataset is projected to include only quasi-identifiers, every record appears at least $k$ times.
\end{definition}

The basic idea behind $k$-anonymity is to generalize 
certain identifying attributes of individuals in the dataset such 
that each group of individuals with similar characteristics contains 
at least $k$ individuals. An outline of the algorithm is 
provided below (a formal description is given in Appendix~\ref{app:sec:DA_alg}):

\begin{enumerate}[leftmargin=*, parsep=0pt, itemsep=0pt, topsep=2pt]
    \item define a \emph{hierarchy} $H$ for each quasi-identifier;
    \item constructs a histogram that lists the number of records for each combination of quasi-identifiers;
    \item suppress the combinations in the generalization histogram that have fewer than $k$ instances.
    \item release the resulting histogram $\tilde{\bm{x}}(D, H)$.
\end{enumerate}
An important observation is that, contrary to the previous methods reviewed, $k$-anonymity produces a privacy-preserving histogram $\tilde{\bm{x}}(D,H)$ in a different space than $\cX$. 
This important observation will be relevant in the error analysis. 
It additionally requires access to quasi-identifier attributes as well as a generalization histogram. 

Figure~\ref{fig:scheme}(e) illustrates the application of $k$-anonymity with $k\!=\!2$ to the histogram of Figure~\ref{fig:scheme}(b), using a
generalization hierarchy grouping Males and Females into a single attribute. 
Despite being widely adopted to publish statistics and medical data, k-anonymity does not prevent re-identification attacks that exploit external public data \citep{li2011provably}. 


\section{DA Analysis Roadmap}
\label{sec:roadmap}

This section outlines the methodology followed in the rest of the
paper.
Section~\ref{sec:privacy_analysis} presents DP counterparts to traditional DA systems, including cell suppression, swapping, and k-anonymity. It aims to show that these DP counterparts preserve the main characteristics of the original mechanisms and provide an analysis of their privacy and errors under a unified privacy setting of histogram data release $\tilde{\bm{x}}(D)$. 
It is important to note that classical DP
algorithms (e.g., Laplace mechanism) and cell suppression, 
make no assumptions about data attributes. In contrast,
swapping relies on the use of quasi-identifiers, 
and $k$-anonymity further requires a generalization hierarchy 
This hierarchy forces k-anonymity to
produce a histogram $\tilde{\bm{x}}(D,H)$ in a different space than
that of $\tilde{\bm{x}}(D)$. While this does not affect the privacy analysis, which
allows for a meaningful comparison across all mechanisms, it
challenges the evaluation of the performance of these techniques.  The
paper addresses this challenge by also presenting a unified empirical
framework for comparing the errors and biases of the various
techniques in terms of the original data space $\cX$. This
necessitates a reconstruction step for k-anonymity, which is outlined
in Appendix~\ref{app:dp_kanonimity}. It is important to recognize that,
while the DP DA mechanisms share many characteristics with their
traditional DA counterparts, {\em they should not be considered as
``noisy'' versions of them}. As a result, the analytical and
experimental results presented may not necessarily show a decrease in
error as the privacy budget increases. In fact, they may even be more
precise than the traditional mechanisms for some privacy budgets.

Next, we present the DP versions of the traditional DA systems and 
their privacy analyses. These analyses specify the value of the $\delta$ parameter for a given value of $\epsilon$. Section~\ref{sec:fairness} analyzes the fairness
 results. Finally, Section~\ref{sec:experiments} presents an
 experimental evaluation on an extract of the American Community Survey (ACS) data \citep{SDNist}.


\section{Privacy and Errors Analysis}
\label{sec:privacy_analysis}

This section presents the first main contribution of the paper. It
introduces differentially private counterparts to the DA presented
earlier and analyzes their privacy guarantees and errors.  The
section starts with a technical lemma that specifies a sufficient
condition for $(\epsilon,\delta)$-DP. The lemma is a critical tool to
derive the privacy guarantees of the randomized versions of the DA
discussed next.

\begin{lemma}\label{lem:dp_cond}
Let $D$, $\ds'$ be datasets such that $D\sim \ds'$, let $S$ be defined as 
\begin{equation*}
        S\coloneqq\left\{\bm{o}~\middle\vert~ \frac{\pr{\cM(\ds)=\bm{o}}}{\pr{\cM(\ds')=\bm{o}}}\leq\exp(\epsilon) \right\}\,
\end{equation*}
and let $S^\complement$ denote the complement set of $S$. If   
\begin{equation*}
\pr{\cM(\ds)\in S^\complement}\leq \delta\,,
\end{equation*}
then mechanism $\cM$ is $(\epsilon,\delta)$-differentially private.
\end{lemma}

\subsection{Differentially Private Cell Suppression}
\label{sec:cell_suppression}


While cell suppression protects the privacy of the minorities of the
dataset, it neglects the privacy protection of the majorities and thus
does not satisfy the requirements of differential privacy.  Indeed,
the deterministic nature of this mechanism prevents it from generating
different outputs for two neighboring datasets. An extended discussion
is deferred to Appendix~\ref{app:sec:DA_alg}. To address this issue, the
paper introduces a randomized version of cell suppression, referred
to as \emph{DP cell suppression}. This mechanism, denoted by $\tmech$,
releases a private count for every $i \in [n]$ as follows:
\begin{equation}
\label{eq:DPsuppression}
\tmech(D)_i = \hat{x}_i =
\begin{cases}
x_i & \text{if } x_i+\eta_i \geq k, \\
\nicefrac{k}{2} & \text{otherwise}
\end{cases},
\end{equation}
where $\eta_i \sim \lap{2/\epsilon}$ is an additive noise variable drawn from a 0-centered Laplace distribution with factor $\nicefrac{2}{\epsilon}$ and $k$ is the cell suppression threshold.

Mechanism $\tmech$ has similarities with the \textsl{Sparse Vector
Technique} (SVT) \citep{dwork2014algorithmic} which, given a sequence
of queries and a real-valued threshold, outputs a vector indicating
whether each (noisy) query answer is above or below the corresponding
(noisy) threshold.  However, there are three fundamental differences:
(1) $\tmech$, does not perturb the threshold value $k$; (2) it
generates numeric outputs in contrast to binary outputs
of \textsl{SVT}; and (3) it reports true counts rather than noisy
counts, as long as the noisy counts are above the threshold (first
condition of Equation~\eqref{eq:DPsuppression}).

Figure~\ref{fig:error}(left) reports the empirical errors of $\tmech$
for several threshold values $k$ ($x$-axis) and $\epsilon$ parameters.
The errors are given for the ACS Massachusetts dataset \citep{SDNist} (described in
details in Appendix~\ref{app:datasets}): they report the $\ell_1$
distances $\| \tilde{\bm{x}} - \bm{x}\|_1$ between the
histograms of the cell suppression and its DP counterpart. Notice
how close the errors incurred by $\tmech$ are with respect to the
original mechanism. 
This is important as it enables a meaningful comparison of $\tmech$ and other DP mechanisms, since
$\tmech$ has a similar bias as the traditional cell suppression that is
currently widely adopted by statistical agencies and organizations.

\paragraph{Privacy Analysis.}
The next theorem reports the privacy guarantee provided by $\tmech$.
\begin{figure*}[!t]
\centering
\includegraphics[width=0.3\textwidth]{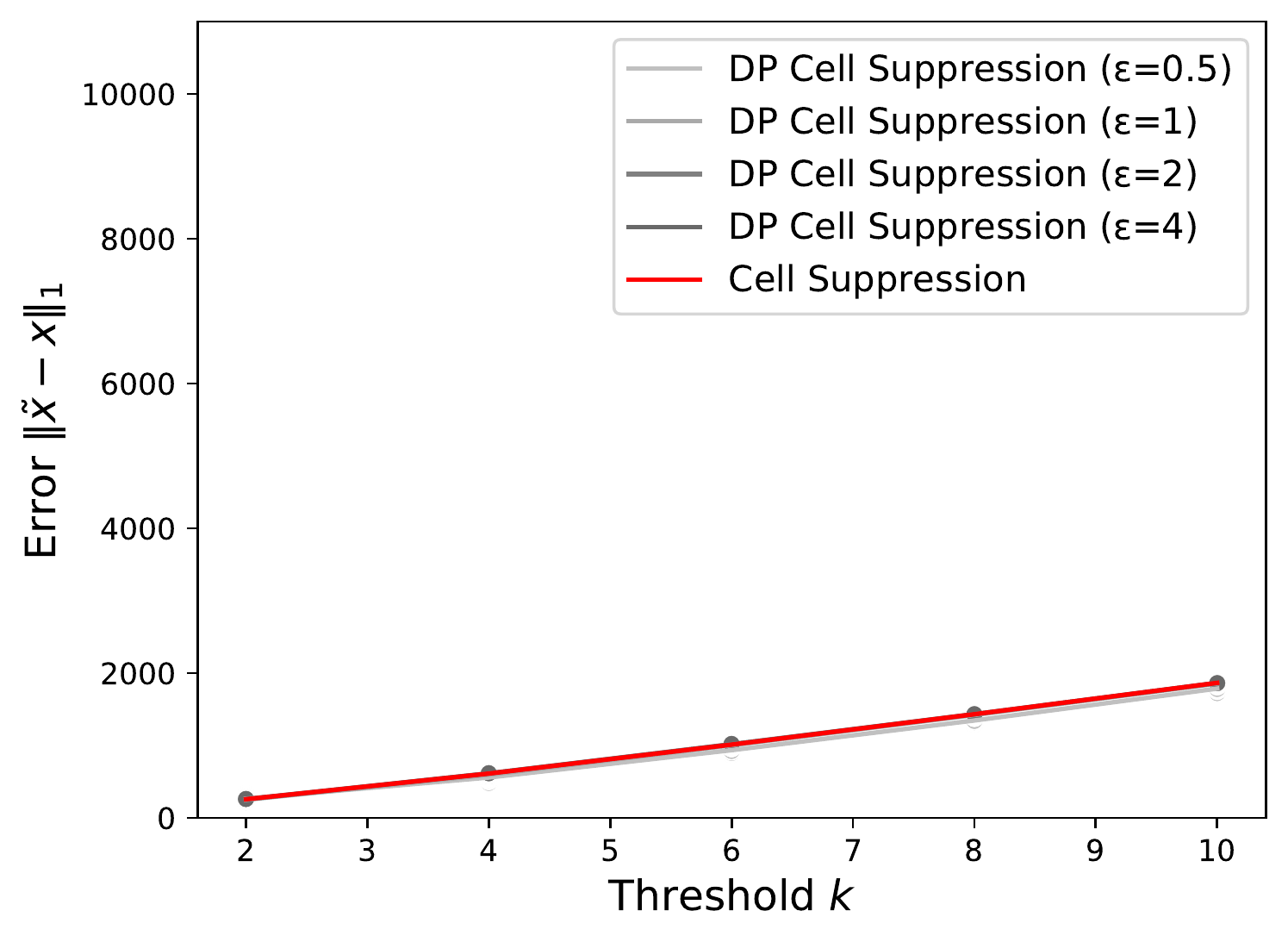} 
\includegraphics[width=0.37 \textwidth]{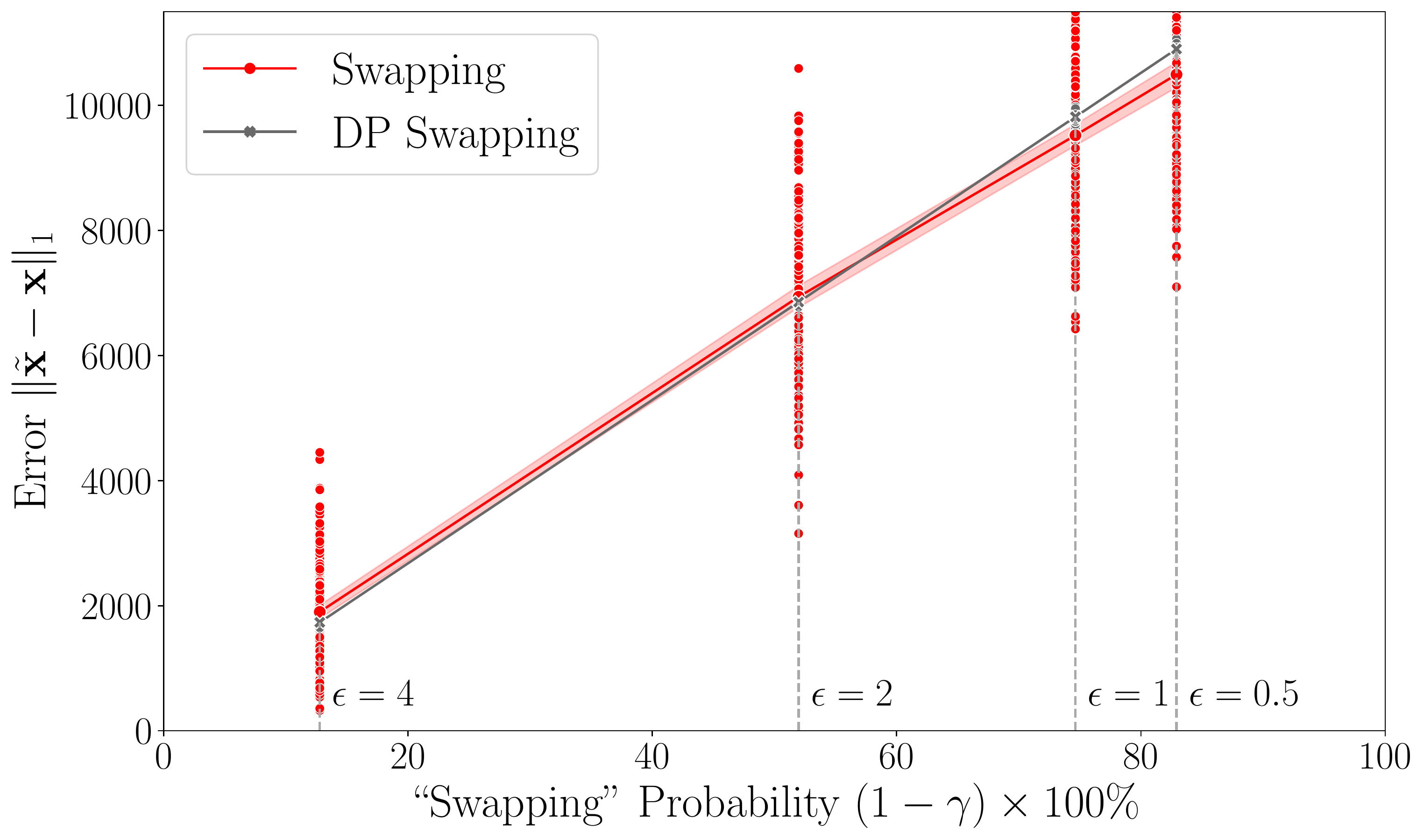}
\includegraphics[width=0.3\textwidth]{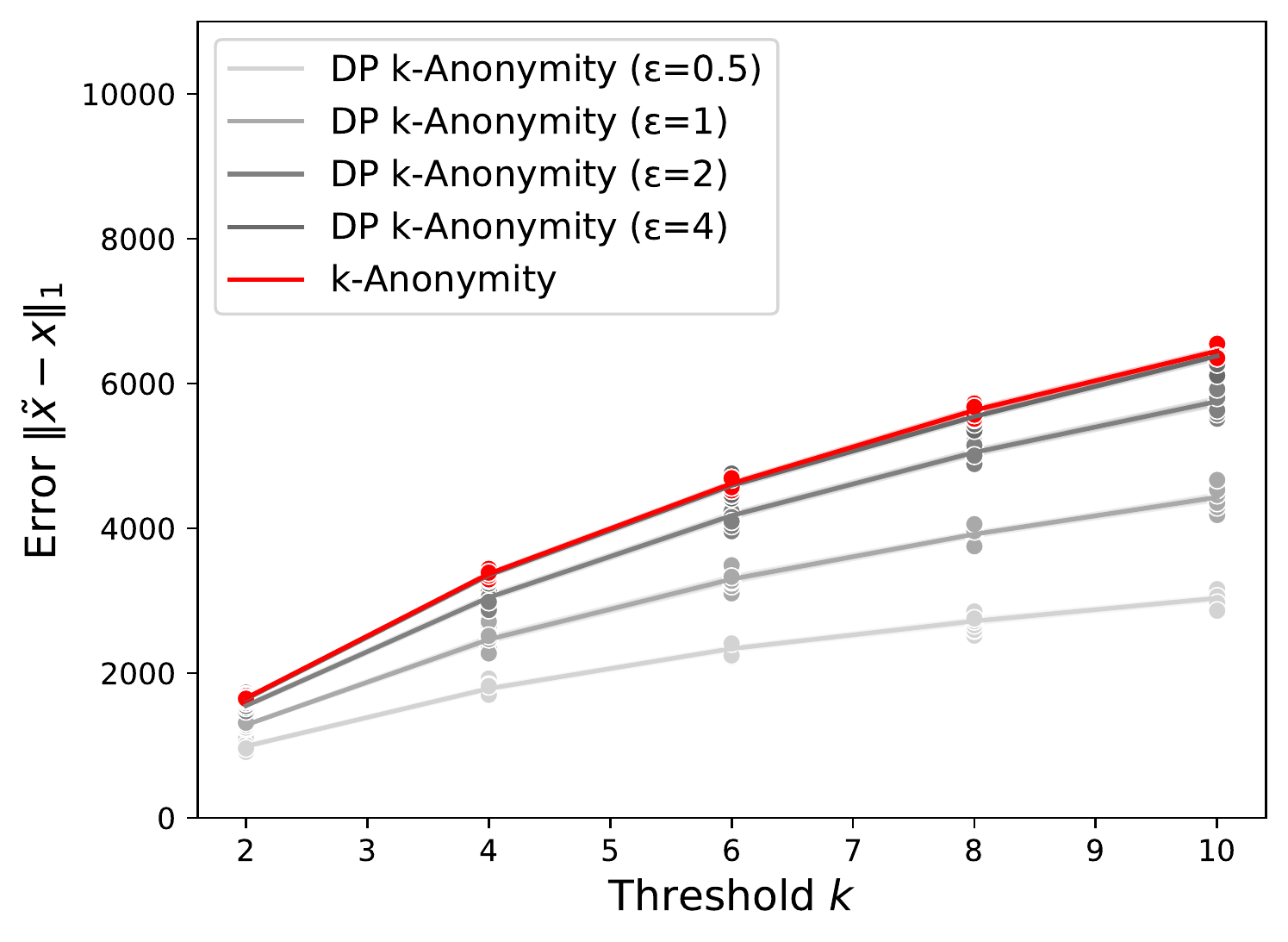}
\caption{MA ACS dataset: Errors $\Vert \tilde{\bm{x}}- \bm{x}\Vert_1$ for cell suppression (left), 
swapping (center) and $k$-anonymity (right) and their differentially 
private counterparts (average of 200 repetitions).}
\label{fig:error}
\end{figure*}

\begin{theorem}\label{thm:sup_dp_param}
Given a value $\epsilon > 0$ and a threshold $\thresh<\bound$, mechanism $\tmech$ is $(\epsilon,\delta)$-differentially private with
\[
        \delta= 1-\frac{1}{4}\exp\left(-\epsilon(\bound-\thresh)\right)\,,
\]
    where $B$ is a bound on the histogram entries.
\end{theorem}

\paragraph{Error Analysis.}
Having examined privacy, the paper shows how close the histograms $\tilde{\bm{x}}$ returned by $\tmech$
are to the original histogram $\bm{x}$. The
error analysis focuses on the statistical bias which, for each entry $i\in[n]$, 
can be expressed as
\begin{equation*}
    \resizebox{.99\linewidth}{!}{$
            \displaystyle
            \biasi{\tmech}{i} =\EE{}{\tmech(\ds)_i}-x_i
    =\left(\frac{\thresh}{2}-x_i\right)\cdot \pr{x_i+\eta_i<\thresh}.
        $}
\end{equation*}%

\noindent
Observe that the error merely takes place when the noisy count is
below the threshold $k$ and is quantified as the difference between
half of the threshold and the true count. Therefore, the following
theorem relates the errors associated with $\tmech$ with the 
probabilities of noisy counts being below the
threshold, and the differences between half of the threshold and the
counts of the original histogram.

\begin{theorem}\label{prop:sup_bias_bound}
The statistical bias of the DP cell suppression mechanism $\tmech$ can be bounded as follows,
    \begin{equation*}
        \lVert \bias{\tmech}\rVert_1\leq \left\lVert \nicefrac{\thresh}{2}\cdot \bm{1}_n - \bm{x}\right\rVert_2 \cdot 
           \left\Vert \bm{p}\right\rVert_2\,,
    \end{equation*}
    where $\bm{p}$ is a shorthand for the vector 
    \begin{equation}\label{eq:p_vector}
        \bm{p}\coloneqq\left[
        \begin{matrix}
         \pr{x_1+\eta_1<\thresh}&\dots&\pr{x_n+\eta_n<\thresh}
        \end{matrix}
        \right].
    \end{equation}
\end{theorem}

\noindent

\subsection{Differentially Private Swapping}\label{subsec:dpswap}
Despite its randomized nature, the swapping mechanism fails to meet
the requirements of differential privacy. To illustrate its failure,
let us take a look at an instance of two neighboring datasets $\ds$
and $\ds'$. Suppose that $\ds'$ has a record, say $\bm{a}_1$, which
does not match any record in $\ds$ for any attribute $A\in Q$. No
matter how swapping is performed, it cannot generate a record
$\bm{a}_1$ from the input dataset $\ds$.

To obtain a DP counterpart to swapping, {\em it is thus critical to
reason about the universe of quasi-identifiers, not simply the set of
quasi-identifiers present in the database.} Let $\qiuniverse
= \{\bm{q}_1,\dots,\bm{q}_{\qiattrs}\}$ denote the data universe of
quasi-identifiers. Instead of swapping quasi-identifiers, the
mechanism will randomly choose some quasi-identifiers from
$\qiuniverse$.  
The mechanism, referred to as DP swapping and denoted
by $\cM_{\text{SW}}$, works as follows: for every $\bm{a}_i \in \cX$,
consider the pair $(\bm{a}_i, x_i)$ denoting the tuple and its
associated count in the histogram $\bm{x}$. $\cM_{\text{SW}}$ defines
$\tilde{\bm{a}}_i[N] = \bm{a}_i[N]$ and
\begin{equation}
\label{eq:dp_swap_def}
\tilde{\bm{a}}_i[Q] =
\begin{cases}
  \bm{a}_i[Q] & \text{w.p. } \gamma = \frac{\exp(\epsilon)}{\exp(\epsilon)+n_Q-1}, \\
\text{Uniform}(\qiuniverse\setminus \bm{a}_i[Q]) & \text{w.p. } 1 - \gamma
\end{cases}
\end{equation}
where $\text{Uniform}(C)$ denotes the uniform probability over the
event space $C$.  The result of the step above may create multiple
entries $(\tilde{\bm{a}}_i, x_i)$ and $(\tilde{\bm{a}}_j, x_j)$ with
$\tilde{\bm{a}}_i = \tilde{\bm{a}}_j$. The procedure collapses all
such tuples by summing the various $x_i$ and $x_j$. The induced
sub-histogram is then extended to a histogram $\tilde{\bm{x}}$.  
Notice that $\cM_{\text{SW}}$ only modifies quasi-identifiers and produces
a private histogram $\tilde{\bm{x}}(\tilde{\ds})$, similarly to what 
done by the original swapping algorithm.

Figure~\ref{fig:error} (center) compares the $\ell_1$ distances
$\| \tilde{\bm{x}} - \bm{x}\|_1$ between the histograms generated by
$\prammech$ and its traditional counterpart for various amounts of
rows swapped (in \%) and parameters $\epsilon$. Once again, observe how close the
errors of the two mechanisms are.

\paragraph{Privacy analysis.}
Recall that differential privacy protects the disclosure of any individual user participating in the dataset. On the other hand, DP swapping operates at the level of a histogram count and it does so in way which impedes an analysis relying on pure $(\epsilon, 0)$-DP. The privacy analysis of $\cM_{\text{SW}}$ is reported in the following theorem.
\begin{theorem}\label{thm:new_pram_dp_param}
    For a given $\epsilon > 0$, the $\cM_{\text{SW}}$ algorithm is $\left(\epsilon,\delta\right)$-DP 
    with $\delta$ given by
    \begin{equation*}
        1-\frac{1-\gamma^2}{\qiattrs-1}-\left(\frac{1-\gamma}{\qiattrs-1}\right)^2,
    \end{equation*}
    with $\gamma$ defined in Equation \eqref{eq:dp_swap_def} and $\qiattrs=\vert\qiuniverse\vert$, 
\end{theorem}

\paragraph{Error analysis.\!\!}
Next we discuss how close the errors of the histograms 
returned by swapping and its DP counterparts are.

\begin{proposition}\label{prop:pram_bias_exp}
The bias associated with each element $i\in[n]$ of the DP swapping 
histogram can be expressed as
\begin{equation*}
    \biasi{\prammech}{i}=\frac{\sum_{j\in \cI_i} x_{j} - \qiattrs\cdot x_{i}}{\exp(\epsilon)+\qiattrs-1},
\end{equation*} 
with the index set $\cI_i$ collecting all the elements of the data universe $\cX=\left\{\bm{a}_j\mid j\in [n]\right\}$, which share the same non-quasi-identifiers
$\nqi$ with $\bm{a}_i$, i.e.,
\begin{equation}\label{eq:index_set}
    \cI_i \coloneqq\left\{j\in[n]~\middle\vert~ \bm{a}_j[\nqi]=\bm{a}_i[\nqi]\right\}.
\end{equation}
\end{proposition}

\begin{theorem}
The statistical bias of the DP swapping mechanism $\prammech$ can be expressed as follows,
    \begin{equation*}
        \norm{\bias{\prammech}}_1=\frac{\qiattrs}{\exp(\epsilon)+\qiattrs-1}\sum_{i=1}^{n}\mad{\bm{x}_{\cI_i}},
    \end{equation*}
    where $\cI_i$ is an index set defined in Equation \eqref{eq:index_set} and
    $\bm{x}_{\cI_i}$ is the reduced histogram consisting of the count $x_j$ for any $j\in\cI_i$.
   Additionally,
    $\mad{\bm{x}_{\cI_i}}$ is the mean absolute deviation of the histogram $\bm{x}_{\cI_i}$, i.e.,
    \begin{equation*}
        \mad{\bm{x}_{\cI_i}}\coloneqq\frac{1}{\qiattrs}\sum_{j\in \cI_i} \left\vert x_{j} - \frac{\sum_{l\in\cI_i} x_{l}}{\qiattrs}\right\vert.
    \end{equation*}
\end{theorem}

\subsection{Differentially Private $k$-anonymity}

Like cell suppression, $k$-anonymity is a deterministic algorithm,
which cannot produce outputs satisfying DP.  In recent years, however, there
were several attempts to integrate $k$-anonymity and differential
privacy. In particular, \citet{li2011provably} proposed a mechanism
that applies the generalization histogram and cell suppression with
parameter $k$ to a subsampled version of the original dataset. This
paper proposes a generalization of this mechanism.

\begin{figure*}[!t]
\centering
\includegraphics[width=0.3\textwidth]{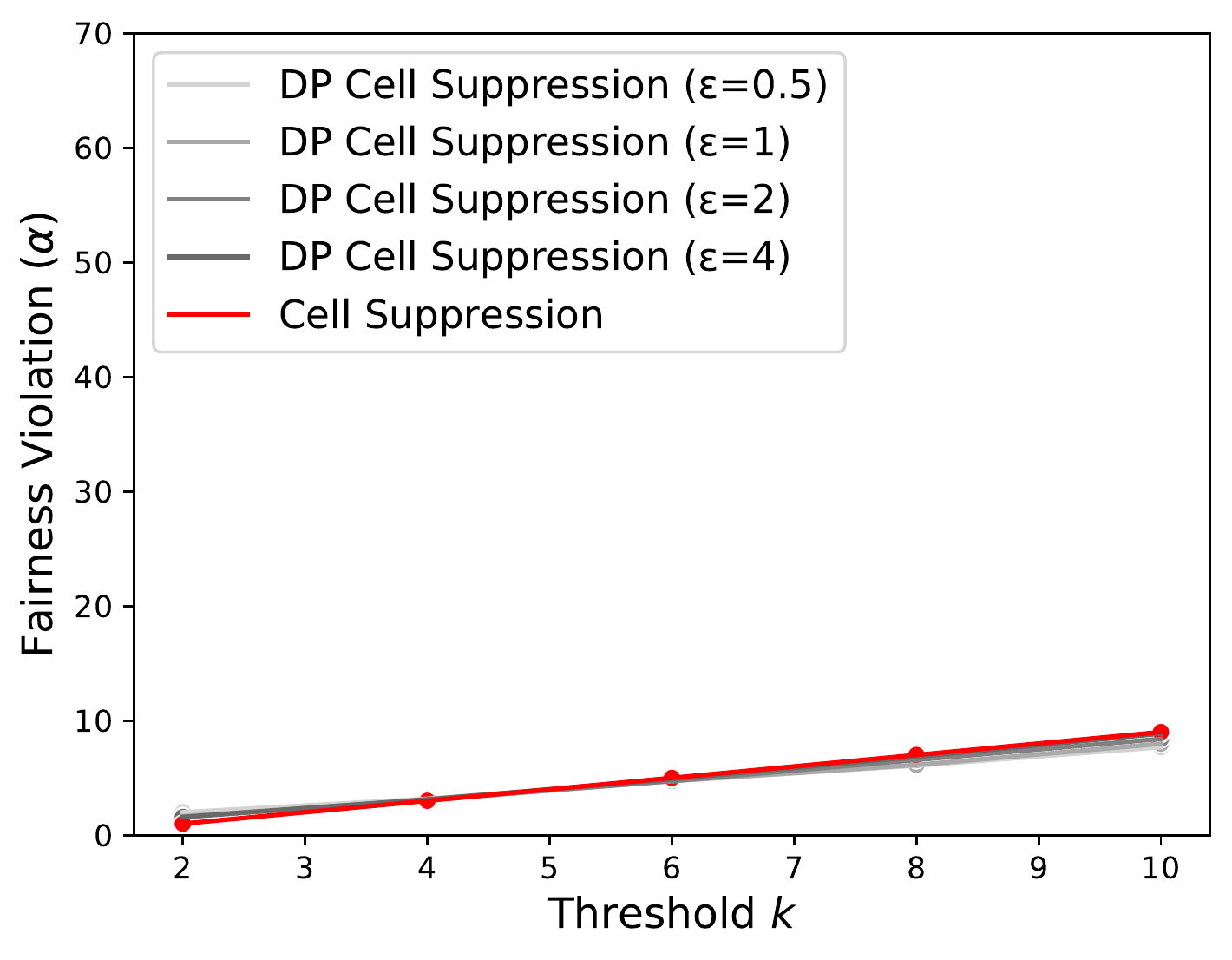}
\includegraphics[width=0.343\textwidth]{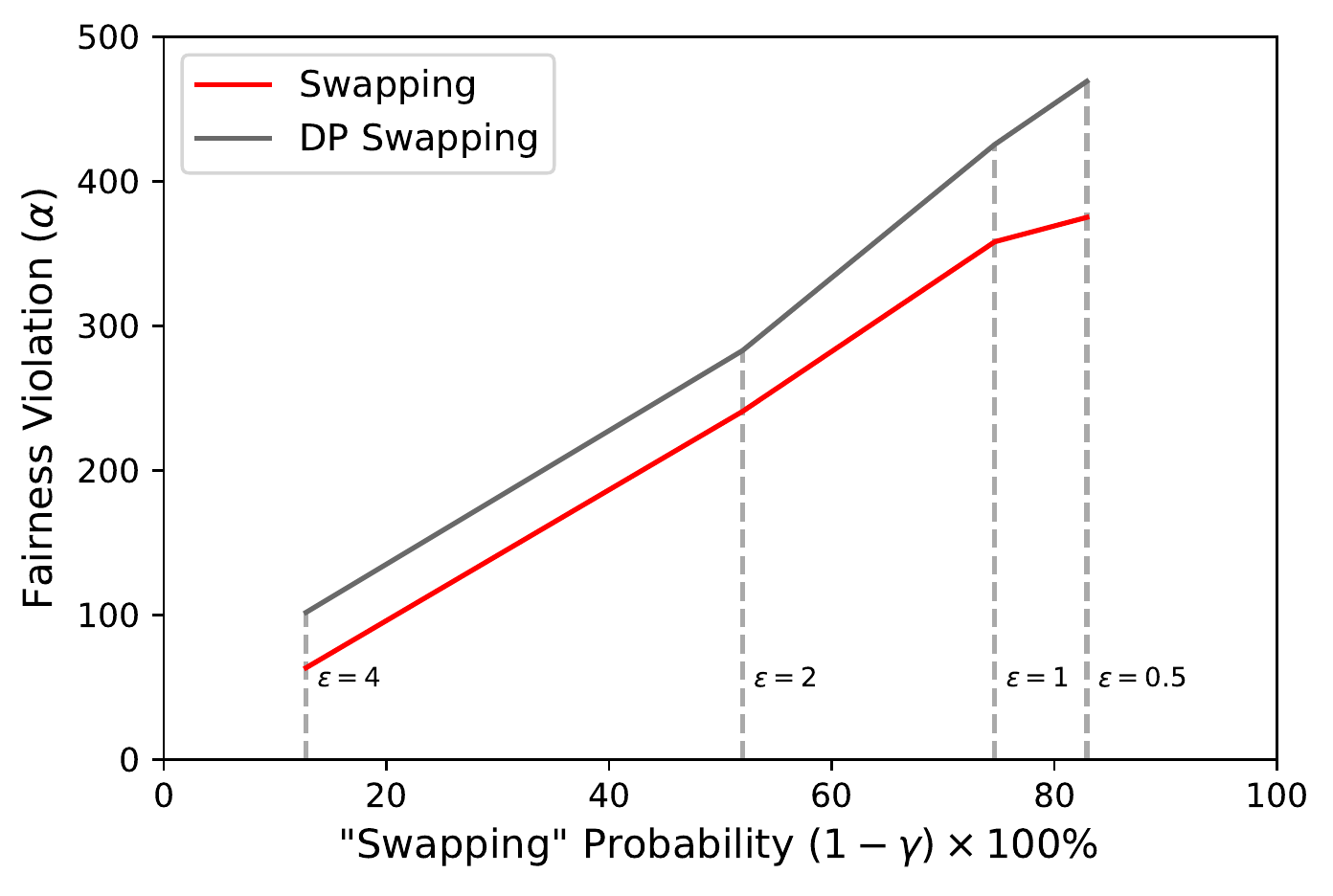}
\includegraphics[width=0.3\textwidth]{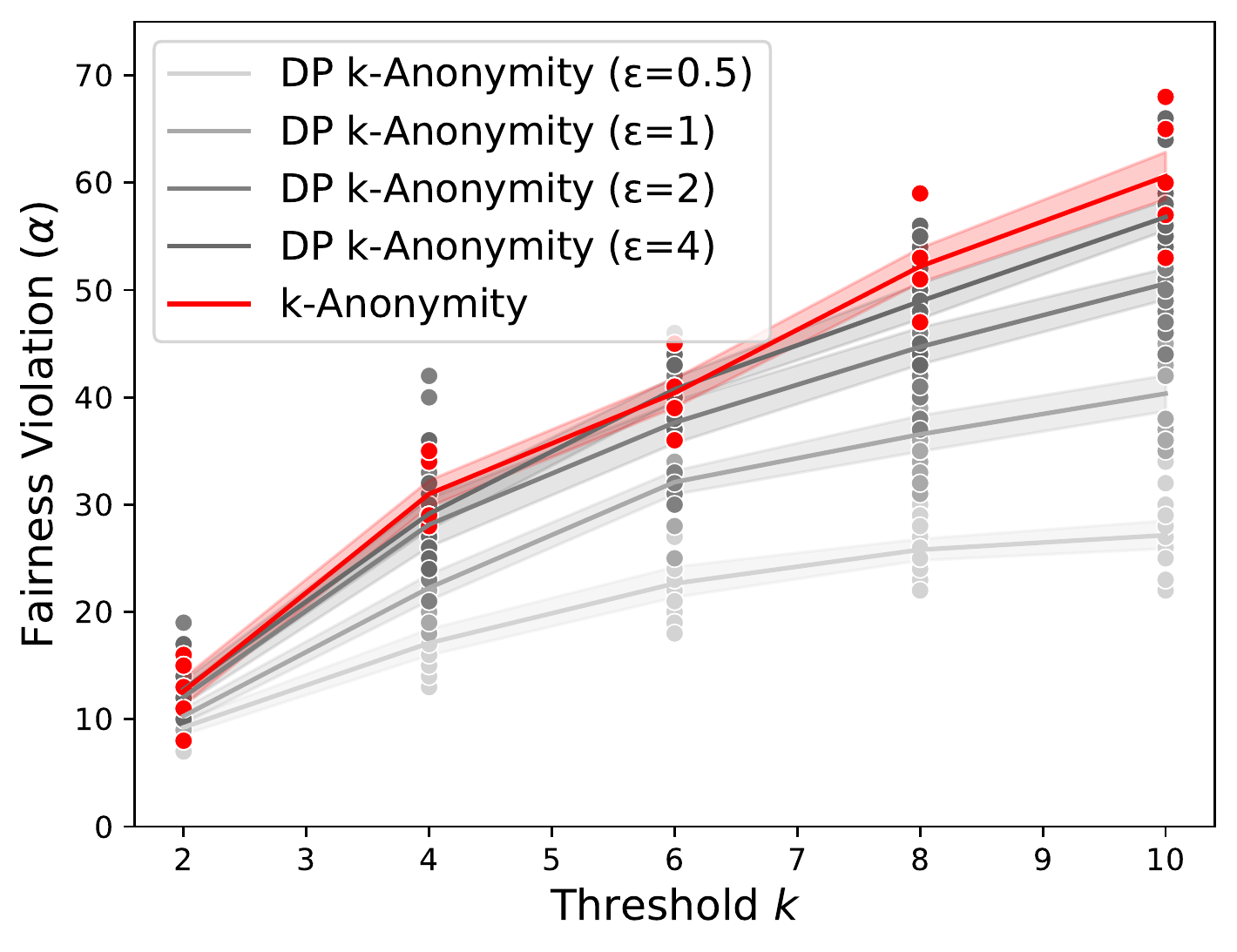}
\caption{MA ACS dataset: Fairness values $\alpha$ for cell suppression (left), 
swapping (center) and $k$-anonymity (right) and their differentially 
private counterparts (average of 200 repetitions).}
\label{fig:fairness}
\end{figure*}

\paragraph{DP $k$-anonymity algorithm.}
The paper proposes a simple modification to the $k$-anonymity 
algorithm presented in the previous section. The mechanism, 
called DP-$k$-anonymity and denoted by $\cM_{\text{KA}}$, operates in two steps:
\begin{enumerate}[leftmargin=*, parsep=2pt, itemsep=0pt, topsep=2pt]
    \item Produce a subsampled version $D_\beta$ of the  dataset $D$ in which each row is retained with probability $\beta \!\in\! (0,1)$.
    \item Apply the classical $k$-anonymity algorithm on $\ds_\beta$.
\end{enumerate}

\noindent
Contrary to \citep{li2011provably}, $\cM_{\text{KA}}$ makes it possible
to use a generalization hierarchy for merging cells, like in its
deterministic counterpart.  Figure~\ref{fig:error} (right) reports the
empirical errors of $\cM_{\text{KA}}$ for several values $k$ (x-axis);
The values of $\beta$ are implied by the choice of the privacy
parameter $\epsilon$ (see Theorem~\ref{thm:beta_dp_param}).  The
figure reports the $\ell_1$ distances between the histograms (in the
original space $\cX$) reconstructed from the generalized DP
k-anonymized and the original histogram $\bm{x}(\ds)$ as well as those
derived via its deterministic counterpart. A description of
the reconstruction step adopted is provided in Appendix~\ref{app:sec:DA_alg}.
Once again, notice that the errors incurred by the DP and
deterministic $k$-anonymity counterparts are very close to each other:
in fact, the DP versions improve upon the deterministic mechanism for
small values of $\epsilon$ as an artifact of the sampling procedure 
which considers fewer records.

\paragraph{Privacy Analysis.}
The next result generalizes \citep{li2011provably} and reports the
privacy guarantees provided by $\cM_{\text{KA}}$.
\begin{theorem}\label{thm:beta_dp_param}
For a given $k\!>\!0$ and sampling probability $\beta \!\in\! (0,1)$, $\cM_{\text{KA}}$ satisfies
$(\epsilon,\delta)$-DP for $\delta \!=\! d(k, \beta, \epsilon)$, where the
function $d$ is defined as
\begin{equation}\label{eq:beta_mech_delta_def}
    d(k, \beta, \epsilon)= 1-\min_{w\in [B]}~
       \left(\sum_{j=0}^{\nu}f(j;w,\beta)\right)^2\,,
\end{equation}
$\nu$ is a shorthand for $\lfloor(1-\exp(-\epsilon))w\rfloor$
and $f$ represents the probability mass function of the binomial
random variable,
\begin{equation}\label{eq:binom_pmf}
f(j;w,\beta)=\binom{w}{j}\beta^w(1-\beta)^{w-j}\,,\quad \forall~j\in
[w], \end{equation}
\end{theorem}

\paragraph{Error Analysis.\!\!}
The goal of the analysis is again to show that the DP version of
$k$-anonymity is close, in errors, to its deterministic counterpart.
Recall that, contrary to cell suppression and swapping, $k$-anonymity
does not return a histogram in the same space of the
attribute universe, due to its application of the generalization
hierarchy.  As a consequence, the output of $k$-anonymity consists of
the counts of individual records with respect to the data universe
$\cX_H$ of the generalization hierarchy. Because of this critical
difference, the rest of this section analyzes the mechanism errors by
bounding whether a count of the histogram $\bm{x}(\ds_\beta)$
(produced in step 1) is merged by the $k$-anonymization step (step~2).

\noindent
Let $\bmech$ denote the binary vector
\begin{equation*}
    \bmech\left(\ds_\beta\right)\coloneqq
    \left[
    \begin{matrix}
         \bm{1}\left\{x_1\left(\ds_\beta\right)<k\right\}&\dots&\bm{1}\left\{x_n\left(\ds_\beta\right)<k\right\}
    \end{matrix}
   \right].
\end{equation*}
The error analysis focuses on statistical bias regarding whether a
count would be merged by the generalization hierarchy, i.e., the
difference between $\bmech\left(\ds_\beta\right)$ and $\bmech(\ds)$:
\begin{align*}
     \bias{\bmech}
    &=\EE{}{\bmech\left(\ds_\beta\right)}-\bmech(\ds)\\
    &=\left[
    \begin{matrix}
        \mathbb{E}\left[\bm{1}\left\{x_1\left(\ds_\beta\right)<k\right\}\right]
        -\bm{1}\left\{x_1<k\right\}\\
        \vdots\\
         \mathbb{E}\left[\bm{1}\left\{x_n\left(\ds_\beta\right)<k\right\}\right]-\bm{1}\left\{x_n<k\right\}
    \end{matrix}
    \right]^\top.
\end{align*}

\noindent
Next, we establish the equivalence between the count
$x_i(\ds_\beta)$ and a binomial random variable $B(x_i, \beta)$ and
presents the mathematical expressions characterizing the bias.
\begin{theorem}
The statistical bias associated of the DP $k$-anonymity mechanism $\bmech$ 
ca be expressed as, each $i\in[n]$,
    \begin{align*}
        &\biasi{\bmech}{i}=\EE{}{\bm{1}\left\{x_i\left(\ds_{\beta}\right)<k\right\}}-\bm{1}\left\{x_i <k\right\}\\
        =~&\begin{cases}
            \sum_{j=x_i-k+1}^{x_i} \binom{x_i}{j} \beta^{x_i-j}(1-\beta)^j\,, &x_i\geq k,\\
            0, &\mathrm{otherwise.}
        \end{cases}
    \end{align*}
\end{theorem}

\section{Fairness analysis }
\label{sec:fairness}

The second main contribution of this paper is an analysis of the
fairness of various differentially private DA algorithms, compared to
traditional differential privacy. The definition of fairness used in
this paper (Definition~\ref{def:a-fair}) is the maximum difference in
biases in the privacy-preserving histograms.  It should be noted that
the bias of the DP $k$-anonymity algorithm, which utilizes a
generalization histogram, is examined in a different context than the
one of the other mechanisms. Thus, the paper specifically focuses on
an analytical comparison of the fairness of the DP $k$-anonymity
algorithm.

The next result quantifies the unfairness of the DP cell suppression and swapping, along with the Laplace mechanism.

\begin{theorem}[$\alpha$-fairness for $\cM_{\text{CS}}$]\label{prop:sup_fair_bound}
  The DP cell suppression algorithm is $\acs$-fair with $\acs$ given by
    \begin{equation*}
        \left(x_n-x_1\right)p_1+\max\left\{\left\vert\frac{\thresh}{2}-x_1\right\vert,\left\vert\frac{\thresh}{2}-x_n\right\vert\right\}(p_1-p_n),
    \end{equation*}
    where $p_1$ and $p_n$ are the first and last entries of $\bm{p}$ defined in Equation \eqref{eq:p_vector} respectively.
\end{theorem}

\begin{theorem}[$\alpha$-fairness for $\cM_{\text{SW}}$]
\label{prop:pram_fair_bound}
    The DP swapping algorithm $\cM_{\text{SW}}$ is
     $\asw$-fair with $\asw$ given by
    \begin{equation*}
       \frac{2\qiattrs\norm{\bm{x}}_{\rightleftharpoons}}{\exp(\epsilon)+\qiattrs-1}=
       \frac{2\qiattrs}{\exp(\epsilon)+\qiattrs-1}(x_n-x_1).
    \end{equation*}
\end{theorem}

\begin{theorem}[$\alpha$-fairness for $\cM_{\text{Lap}}$]
\label{prop:lap_fair_bound}
The Laplace mechanism $\cM_{\text{Lap}}$ is $\alap$-fair with $\alap$ given by
    \begin{equation*}
        \frac{\exp\left(-\epsilon x_1/2\right)}{2}\norm{\bm{x}}_{\rightleftharpoons}=
       \frac{\exp\left(-\epsilon x_1/2\right)}{2}\left(x_n-x_1\right).
    \end{equation*}
\end{theorem}

Figure~\ref{fig:fairness} illustrates the fairness violations values,
represented by the value of $\alpha$, for cell suppression, swapping,
and $k$-anonymity, as well as their differentially private
counterparts, for various privacy parameters $\epsilon$ and values of
$k$ (for cell suppression and $k$-anonymity) or percentage of rows
swapped (for swapping). It can be observed that the fairness
violations of the differentially private mechanisms are comparable (or
better) to those of their traditional counterparts. This is
particularly noteworthy as the privacy parameter $\epsilon$
increases. As previously mentioned, it is important to remember that
the differentially private mechanisms are not ``noisy'' versions of
their traditional counterparts; rather they are conceptually similar
mechanisms. Consequently, they may exhibit lower fairness violations
compared to their traditional counterparts, as seen for DP $k$-anonymity.

The following theorem is the third key result of this paper. {\em It
proves the superiority of the Laplace mechanism over DP cell suppression and swapping in terms of fairness errors.} 

\begin{theorem}\label{thm:fair-comp}
    Suppose that the minimum count of the original histogram $\bm{x}(\ds)$
    is between $2$ and the threshold $\thresh$, i.e., $2\leq x_1\leq \thresh$. Then,
    the fairness error associated with the Laplace mechanism is not
    greater than that of the DP cell suppression or DP swapping mechanism, namely,
    \begin{equation*}
        \alap\leq \acs\quad\text{and}\quad\alap\leq \asw.
    \end{equation*}
\end{theorem}

It is worth noting that $k$-anonymity operates in a different space
from the original histogram space, thus a theoretical comparison
between the Laplace mechanism and DP $k$-anonymity is not feasible. However, the paper next presents empirical evidence that the Laplace mechanism has 
a significant advantage over DP $k$-anonymity as well.

\section{Experimental Evaluation}
\label{sec:experiments}

This study assesses the performance of the DP variants of 
traditional DA mechanisms and compares them with two key DP
mechanisms, the Laplace and the Discrete Gaussian Mechanisms, reviewed
in Section~\ref{sec:das}. The experiments use the ACS 2019 IPUMS
datasets for Massachusetts, Texas, and Outlier \citep{SDNist}. All the experiments report the average of 200 repetitions. Results for the latter two datasets are included in the appendix as their trends are similar 
to the former. The appendix also includes a more extensive description of the
dataset and experimental settings. This section focuses on evaluating
the mechanisms in two settings: data release and classification.

\begin{table}[t]
\small
\centering
\resizebox{0.95\linewidth}{!}
{
    \begin{tabular}{c|l|r|r|r}
    \toprule
    $\epsilon$ & Mechanism & \multicolumn{1}{c|}{$\delta$} & Bias ($\ell_1$ norm) & $\alpha$-fairness \\
    \midrule
    \multirow{5}[2]{*}{0.5} & Laplace & \textbf{0} & \textbf{763.775} & \textbf{3.655} \\
          & Discrete Gaussian & 0.363  & 980.81  & 4.945  \\
          & DP Suppression & 0.999  & 935.525  & 4.345  \\
          & DP Swapping & 0.868 & 10906.79  & 469.015  \\
          & DP $k$-anonymity & 0.878  & 2337.8  & 22.65  \\
    \midrule
    \multirow{5}[2]{*}{1} & Laplace & \textbf{0} & \textbf{342.885} & \textbf{1.845} \\
          & Discrete Gaussian & 0.132  & 659.215  & 3.065  \\
          & DP Suppression & 0.999  & 1003.035  & 4.5  \\
          & DP Swapping & 0.874 & 9859.26  & 425.335  \\
          & DP $k$-anonymity & 0.906  & 3297.4  & 32.1\\
    \midrule
    \multirow{5}[2]{*}{2} & Laplace & \textbf{0} & \textbf{154.78} & \textbf{0.905} \\
          & Discrete Gaussian & 0.017  & 436.925  & 2.2  \\
          & DP Suppression & 0.999  & 1018.335  & 4.72  \\
          & DP Swapping & 0.899 & 6841.19  & 282.73  \\
          & DP $k$-anonymity & 0.981  & 4175.5  & 37.65  \\
    \midrule
    \multirow{5}[2]{*}{4} & Laplace & \textbf{0} & \textbf{67.34} & \textbf{0.465} \\
          & Discrete Gaussian & 3E-4  & 290.715  & 1.495  \\
          & DP Suppression & 0.999  & 1014.63  & 4.92  \\
          & DP Swapping & 0.969 & 1664.63  & 101.645  \\
          & DP $k$-anonymity & 0.999  & 4590.7  & 40.75 \\
    \bottomrule
    \end{tabular}%
}
  \caption{MA dataset data release: Comparison of DP mechanisms in terms of $\delta$, $\ell_1$ norm of the empirical bias and $\alpha$-fairness.  \label{tab:compareDPmethods}}
\end{table}%

\paragraph{Data Release.}

The first task compares datasets reconstructed from histograms
generated by the various DP mechanisms studied. Readers are referred
to Appendix~\ref{app:sec:DA_alg} for details on the reconstruction
algorithms.  Table~\ref{tab:compareDPmethods} assesses the performance
of the DP variants of the traditional DA mechanisms and the Laplace
and the discrete Gaussian mechanisms in terms of errors and fairness
violations. In mechanisms that may produce negative counts, a
simple post-processing projection into the non-negative orthant is applied.

These results are particularly significant: contrary to commonly held
beliefs, they demonstrate that classical DP algorithms not only
provide strong privacy guarantees (see the $\delta$ values), but also
produce histograms that improve over traditional DA mechanisms in
terms of both accuracy (see the bias column) and fairness metrics (see
$\alpha$-fairness column). {\em As a consequence, when agencies desire
to release data sets, it is advisable they consider traditional DP
mechanisms.}

\begin{figure}[!t]
\centering
        \includegraphics[width=0.8\linewidth]{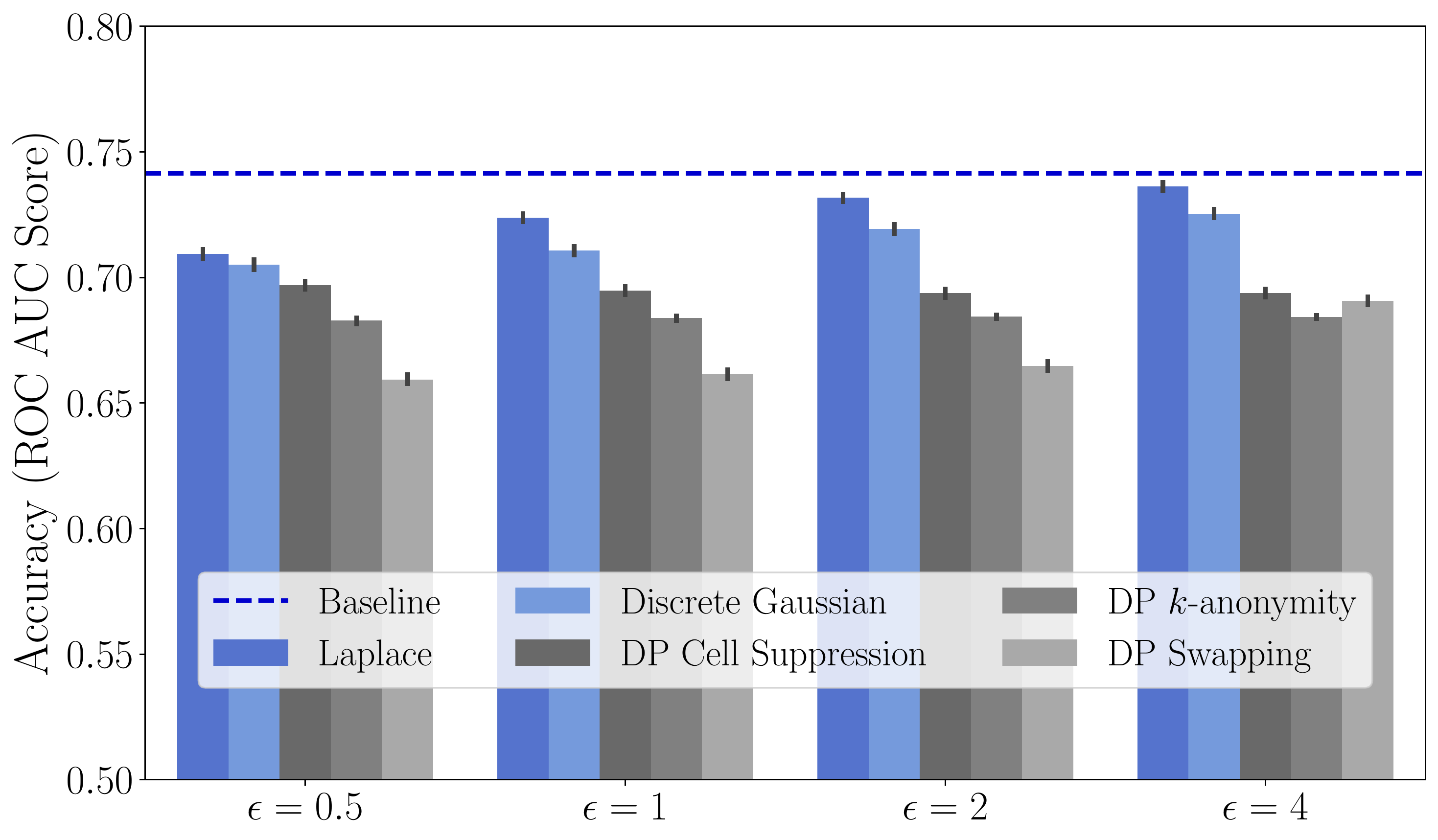}
    \caption{MA dataset: Results for Logistic Regression.}
    \label{fig:Log_errors}
\end{figure}
\paragraph{Classification.} 

It is also important to compare the performance of all the DP
mechanisms in a classification task. The setting employs the private
datasets obtained through a data-release query in order to train a
logistic regression classifier. The task is to predict whether an
individual earns more than \$50,000 per year, and the results in Figure~\ref{fig:Log_errors} are presented in terms of accuracy on the original,
non-private dataset. Observe how the Laplace and discrete Gaussian mechanisms
lead to classifiers with much higher accuracies than classifiers
trained over data produced by other traditional DA
mechanisms. Notably, the classification accuracy of Laplace and discrete 
Gaussian is much closer to that of the baseline method (trained on
non-private datasets) than any other method.  Again, this is
significant: {\em despite their simplicity, these tasks are the
basis for numerous statistical analyses performed routinely by data
agencies and organizations.}

\section{Conclusion}

This paper presented a framework for comparing traditional disclosure
avoidance systems (DA) to differential privacy. It proposed carefully 
randomized versions of three widely adopted traditional DA methods, i.e., 
suppression, swapping, and k-anonymity, and derived $(\epsilon,\delta)$-DP bounds 
for these mechanisms. The paper
also analyzed these DP algorithms empirically and showed
that they are close to their traditional counterparts both in terms of
accuracy and fairness. The DP DA mechanisms were then compared experimentally
with traditional DP mechanisms (i.e., the Laplace or the discrete Gaussian
mechanisms)  on data release and classification tasks. Contrary to popular belief, 
the experimental evaluation showed that classical DP mechanisms may be superior to traditional DA in terms of accuracy and fairness for the same privacy levels. This study has the
potential to impact the way in which data agencies and organizations
approach disclosure avoidance in the future as it provides a framework that enables 
a comparison of the strengths and limitations of traditional DA and differential privacy.

\newpage
\section*{Ethical Statement}
From an ethical standpoint, the study's purpose is not to condone the release of data by agencies that have not fully considered the privacy implications of their actions. The study should not be taken as a means to discredit traditional DA methods. Furthermore, the empirical analysis presented should be understood as specific to the mechanisms and datasets discussed in the study.

It is also important to consider the potential benefits of the study, such as improved accuracy and fairness in data release which may be gained with the adoption of traditional differentially private tools. Additionally, the study has the potential to advance the development of more effective privacy-preserving technologies.

\section*{Acknowledgments}

This research is partially supported by NSF grant 2133169, AI institute 2112533, 
and NSF CAREER Award 2143706. 
Fioretto is also supported by a Google Research Scholar Award and an Amazon Research Award. 
Its views and conclusions are those of the authors only. 
The authors would also like to thank Matt Williams for the helpful discussions 
on swapping and other traditional DA techniques.

\bibliographystyle{named}
\bibliography{ijcai23}

\begin{thebibliography}{}

\bibitem[\protect\citeauthoryear{Canonne \bgroup \em et al.\egroup
  }{2020}]{canonne2020discrete}
Cl{\'e}ment~L Canonne, Gautam Kamath, and Thomas Steinke.
\newblock The discrete gaussian for differential privacy.
\newblock {\em Advances in Neural Information Processing Systems},
  33:15676--15688, 2020.

\bibitem[\protect\citeauthoryear{Clarke and Ledyaev}{1994}]{clarke1994mean}
FH~Clarke and Yu~S Ledyaev.
\newblock Mean value inequalities.
\newblock {\em Proceedings of the American Mathematical Society}, pages
  1075--1083, 1994.

\bibitem[\protect\citeauthoryear{Dalenius and Reiss}{1982}]{dalenius1982data}
Tore Dalenius and Steven~P Reiss.
\newblock Data-swapping: A technique for disclosure control.
\newblock {\em Journal of statistical planning and inference}, 6(1):73--85,
  1982.

\bibitem[\protect\citeauthoryear{Dwork \bgroup \em et al.\egroup
  }{2006}]{Dwork:06}
Cynthia Dwork, Frank McSherry, Kobbi Nissim, and Adam Smith.
\newblock Calibrating noise to sensitivity in private data analysis.
\newblock In {\em Theory of cryptography conference}, pages 265--284. Springer,
  2006.

\bibitem[\protect\citeauthoryear{Dwork \bgroup \em et al.\egroup
  }{2014}]{dwork2014algorithmic}
Cynthia Dwork, Aaron Roth, et~al.
\newblock The algorithmic foundations of differential privacy.
\newblock {\em Foundations and Trends{\textregistered} in Theoretical Computer
  Science}, 9(3--4):211--407, 2014.

\bibitem[\protect\citeauthoryear{Fioretto \bgroup \em et al.\egroup
  }{2022}]{Fioretto:IJCAI22a}
Ferdinando Fioretto, Cuong Tran, Pascal {Van Hentenryck}, and Keyu Zhu.
\newblock Differential privacy and fairness in decisions and learning tasks: A
  survey.
\newblock In {\em In Proceedings of the International Joint Conference on
  Artificial Intelligence ({IJCAI})}, pages 5470--5477, 2022.

\bibitem[\protect\citeauthoryear{Garfinkel \bgroup \em et al.\egroup
  }{2019}]{garfinkel2019understanding}
Simson Garfinkel, John~M Abowd, and Christian Martindale.
\newblock Understanding database reconstruction attacks on public data.
\newblock {\em Communications of the ACM}, 62(3):46--53, 2019.

\bibitem[\protect\citeauthoryear{Kelly \bgroup \em et al.\egroup
  }{1992}]{kelly1992cell}
James~P Kelly, Bruce~L Golden, and Arjang~A Assad.
\newblock Cell suppression: Disclosure protection for sensitive tabular data.
\newblock {\em Networks}, 22(4):397--417, 1992.

\bibitem[\protect\citeauthoryear{Kuppam \bgroup \em et al.\egroup
  }{2019}]{kuppam2019fair}
Satya Kuppam, Ryan McKenna, David Pujol, Michael Hay, Ashwin Machanavajjhala,
  and Gerome Miklau.
\newblock Fair decision making using privacy-protected data.
\newblock {\em arXiv preprint arXiv:1905.12744}, 2019.

\bibitem[\protect\citeauthoryear{LeFevre \bgroup \em et al.\egroup
  }{2006}]{LeFevre2006Mondrian}
K.~LeFevre, D.J. DeWitt, and R.~Ramakrishnan.
\newblock Mondrian multidimensional k-anonymity.
\newblock In {\em 22nd International Conference on Data Engineering (ICDE'06)},
  pages 25--25, 2006.

\bibitem[\protect\citeauthoryear{Li \bgroup \em et al.\egroup
  }{2011}]{li2011provably}
Ninghui Li, Wahbeh~H Qardaji, and Dong Su.
\newblock Provably private data anonymization: Or, k-anonymity meets
  differential privacy.
\newblock {\em CoRR, abs/1101.2604}, 49:55, 2011.

\bibitem[\protect\citeauthoryear{NIST}{2021}]{SDNist}
NIST.
\newblock Sdnist v1.4 beta: Synthetic data report tool. national institute of
  standards and technology, 2021.

\bibitem[\protect\citeauthoryear{Sweeney}{2002}]{sweeney2002k}
Latanya Sweeney.
\newblock k-anonymity: A model for protecting privacy.
\newblock {\em International journal of uncertainty, fuzziness and
  knowledge-based systems}, 10(05):557--570, 2002.

\bibitem[\protect\citeauthoryear{{Tatauranga Aotearoa}}{2020}]{mog}
{Tatauranga Aotearoa}.
\newblock Microdata output guide.
\newblock
  \url{https://www.stats.govt.nz/assets/Methods/Microdata-Output-Guide-2020-v5-Sept22update.pdf},
  2020.

\bibitem[\protect\citeauthoryear{Tran \bgroup \em et al.\egroup
  }{2021}]{Fioretto:IJCAIa}
Cuong Tran, Ferdinando Fioretto, Pascal {Van Hentenryck}, and Zhiyan Yao.
\newblock Decision making with differential privacy under the fairness lens.
\newblock In {\em International Joinmogt Conference on Artificial Intelligence
  ({IJCAI})}, pages 560--566, 2021.

\bibitem[\protect\citeauthoryear{Zhu \bgroup \em et al.\egroup
  }{2022}]{ijcai2022p559}
Keyu Zhu, Ferdinando Fioretto, and Pascal Van~Hentenryck.
\newblock Post-processing of differentially private data: A fairness
  perspective.
\newblock In Lud~De Raedt, editor, {\em Proceedings of the Thirty-First
  International Joint Conference on Artificial Intelligence, {IJCAI-22}}, pages
  4029--4035. International Joint Conferences on Artificial Intelligence
  Organization, 7 2022.
\newblock Main Track.

\end{thebibliography}

\newpage
\appendix

\section{Missing Proofs}

\begin{proof}[Proof of Lemma \ref{lem:dp_cond}]
    For any output $O\subseteq\cR$ and neighboring datasets $\ds, \ds' \in \cX^m$,
    \begin{align*}
        &\pr{\cM(\ds)\in O}\\
        =~&\pr{\cM(\ds)\in \left(O\cap S\right)}+\pr{\cM(\ds)\in \left(O\cap S^{\complement}\right)}\\
        \leq~& \int_{\bm{o}\in \left(O\cap S\right)} \pr{\cM(\ds)=\bm{o}}d\bm{o}+\pr{\cM(\ds)\in S^{\complement}}\\
        \leq~& \exp(\epsilon)\cdot\int_{\bm{o}\in \left(O\cap S\right)} \pr{\cM(\ds')=\bm{o}}d\bm{o}+\delta\\
        \leq~& \exp(\epsilon)\cdot\int_{\bm{o}\in O} \pr{\cM(\ds')=\bm{o}}d\bm{o}+\delta\\
        =~&\pr{\cM(\ds')\in O}+\delta.
    \end{align*}
\end{proof}

\begin{proof}[Proof of Theorem \ref{thm:sup_dp_param}]
   Without loss of generality, the neighboring datasets, $\ds, \ds'\in\cX^m$, are 
   assumed to have different last records
   with $\bm{a}_n$ in $\ds$ and $\bm{a}_{n-1}$ in $\ds'$. It implies that
   the histograms of $\ds$ and $\ds'$ differ in the last two entries, i.e.,
   \begin{align*}
       \bm{x}(\ds')&=
\left[
       \begin{matrix}
           x'_1 & \dots & x'_{n-1} & x'_n
       \end{matrix}
       \right]
       \\
       &=\left[
       \begin{matrix}
           x_1 & \dots & x_{n-1}+1 & x_n - 1
       \end{matrix}
       \right].
   \end{align*}
   Consider the following set
   \begin{equation*}
       \underline{S}=\left\{\bm{o}\in\NN_+^n~\middle\vert~o_{n-1}=o_n=\frac{T}{2}\right\}.
   \end{equation*}
   For any element $\bm{o}\in \underline{S}$, it follows that
   \begin{align}
        &\frac{\pr{\tmech(\ds)=\bm{o}}}{\pr{\tmech(\ds')=\bm{o}}}\nonumber
        \\=~&\frac{\pr{x_{n-1}+\eta_{n-1}<\thresh }}{\pr{x'_{n-1}+\eta'_{n-1}<\thresh }}\cdot
        \frac{\pr{x_{n}+\eta_{n}<\thresh }}{\pr{x'_{n}+\eta'_{n}<\thresh }}\label{eq:thresh_mech_priv_aux_0}\\
        =~&\frac{\int_{-\infty}^{\thresh}\exp\left(-\frac{\epsilon}{2}\left\vert
        v-x_{n-1}\right\vert\right)dv}{\int_{-\infty}^{\thresh}\exp\left(-\frac{\epsilon}{2}\left\vert
        v-x'_{n-1}\right\vert\right)dv}\cdot
        \frac{\int_{-\infty}^{\thresh}\exp\left(-\frac{\epsilon}{2}\left\vert
        v-x_{n}\right\vert\right)dv}{\int_{-\infty}^{\thresh}\exp\left(-\frac{\epsilon}{2}\left\vert
        v-x'_{n}\right\vert\right)dv}\nonumber\\
        \leq~&\exp\left(\frac{\epsilon}{2}\right)\cdot \exp\left(\frac{\epsilon}{2}\right)\label{eq:thresh_mech_priv_aux_1}\\
        =~&\exp(\epsilon)\nonumber,
   \end{align}
   where, in Equation \eqref{eq:thresh_mech_priv_aux_0}, $\eta_{n-1}$, $\eta'_{n-1}$, $\eta_{n}$, and $\eta'_{n}$ are all
   i.i.d. Laplacian random variables with the parameter $2/\epsilon$.
   Besides, Equation \eqref{eq:thresh_mech_priv_aux_1} comes from the triangle inequality.
   This inequality implies that the set $\underline{S}$ is a subset of $S$ in Lemma \ref{lem:dp_cond}.
    As a result, the probability $\pr{\tmech(\ds)\in S^{\complement}}$ can then be
    evaluated as follows
   \begin{align}
      &\pr{\tmech(\ds)\in S^{\complement}}\nonumber\\
      \leq~& \pr{\tmech(\ds)\in \underline{S}^{\complement}}=1-\pr{\tmech(\ds)\in \underline{S}}\nonumber\\
      =~&1-\pr{x_{n-1}+\eta_{n-1}<\thresh }\cdot\pr{x_{n}+\eta_{n}<\thresh }\nonumber\\
      \leq ~& 1-\pr{\bound+\eta_{n-1}<\thresh }\cdot\pr{\bound+\eta_{n}<\thresh }\label{eq:thresh_mech_priv_2}\\
      =~&1-\frac{1}{4}\exp\left(-\epsilon(\bound-\thresh)\right),\label{eq:thresh_mech_priv_3}
   \end{align}
   where Equation \eqref{eq:thresh_mech_priv_2} comes from the fact that the function
   $x\mapsto \pr{x+\eta<\thresh}$ with $\eta\sim\lap{2/\epsilon}$ is decreasing and the histogram
   $\bm{x}(\ds)$ is assumed to be bounded by the constant $\bound$ entrywise.
   Additionally, Equation \eqref{eq:thresh_mech_priv_3} is due to the assumption that $\thresh<\bound$.
   By Lemma \ref{lem:dp_cond}, the privacy guarantee is established for the mechanism $\tmech$,
   which completes the proof here.
   
\end{proof}

\begin{proof}[Proof of Theorem \ref{prop:sup_bias_bound}]
    By the Cauchy-Schwarz inequality,
    \begin{align*}
            \lVert \bias{\tmech}\rVert_1&=
            \sum_{i=1}^n \left\vert\left(\frac{\thresh}{2}-x_i\right)\cdot \pr{x_i+\eta_i<\thresh}\right\vert\\
            &\leq \left\lVert \frac{\thresh}{2}\cdot \bm{1}_n - \bm{x}\right\rVert_2 \cdot 
           \left\Vert \bm{p}\right\rVert_2\,.
        \end{align*}
\end{proof}

\begin{proof}[Proof of Theorem \ref{thm:new_pram_dp_param}]
       Without loss of generality, the neighboring datasets, $\ds, \ds'\in\cX^m$, are 
   assumed to have different last records
   with $\bm{a}_n$ in $\ds$ and $\bm{a}_{n-1}$ in $\ds'$, where $\bm{a}_n[\qi]\neq\bm{a}_{n-1}[\qi]$ and $\bm{a}_n[\nqi]=\bm{a}_{n-1}[\nqi]$.
   Consider the following set
   \begin{equation*}
       S=\left\{\tilde{\bm{x}}\in\NN_+^{n}~\middle\vert~\tilde{\bm{a}}_{n-1}[\qi]=
       \tilde{\bm{a}}_{n}[\qi]
       \right\}\,,
   \end{equation*}
   where $\tilde{\bm{x}}$ is the resulting histogram associated with
   $\{\tilde{\bm{a}}_i\mid i\in[n]\}$ generated by the DP swapping mechanism. It is straightforward to see that, for 
   any $\tilde{\bm{x}}\in S$,
   \begin{equation*}
       \exp\left(-\epsilon\right)\leq\frac{\pr{\prammech(\ds)=\tilde{\bm{x}}}}{\pr{\prammech(\ds')=\tilde{\bm{x}}}}\leq \exp\left(\epsilon\right)\,.
   \end{equation*}
   Then, it follows that
   \begin{align*}
       &\pr{\prammech(\ds)\in S^\complement}\\
       =~&1-\pr{\prammech(\ds)\in S}\\
       =~&1-\sum_{j\in\cI_n}\pr{\tilde{\bm{a}}_{n-1}[\qi]=
       \tilde{\bm{a}}_{n}[\qi]=\bm{a}_j[\qi]}\\
       =~&1-2\frac{\gamma(1-\gamma)}{\qiattrs-1}-(\nqiattrs -2)\left(\frac{1-\gamma}{\qiattrs-1}\right)^2\\
       =~&1-\frac{1-\gamma^2}{\qiattrs-1}-\left(\frac{1-\gamma}{\qiattrs-1}\right)^2\,.
   \end{align*}
   By Lemma \ref{lem:dp_cond}, it provides the privacy guarantee 
   for the DP swapping mechanism.
\end{proof}

\begin{proof}[Proof of Theorem \ref{thm:beta_dp_param}]
    In order to derive the privacy guarantee for the DP $k$-anonymity mechanism,
    it suffices to show that the sample dataset $\ds_{\beta}$
    satisfies $(\epsilon,\delta)$-differential privacy with $\delta$
    defined in Equation \eqref{eq:beta_mech_delta_def} because 
    of the significant property, known as post-processing immunity 
    \citep{dwork2014algorithmic}.

Prior to analysis, let $\cM_{\beta}$ denote the randomized mechanism which generates the sample dataset.
    Without loss of generality, assume that the neighboring datasets, $\ds, \ds'\in\cX^m$, have different last records
   with $\bm{a}_n$ in $\ds$ and $\bm{a}_{n-1}$ in $\ds'$. It implies that
   the histograms of $\ds$ and $\ds'$ differ in the last two entries, i.e.,
   \begin{align*}
       \bm{x}(\ds')&=
\left[
       \begin{matrix}
           x'_1 & \dots & x'_{n-1} & x'_n
       \end{matrix}
       \right]
       \\
       &=\left[
       \begin{matrix}
           x_1 & \dots & x_{n-1}+1 & x_n - 1
       \end{matrix}
       \right]\,.
   \end{align*}
   Consider the set 
   \begin{multline*}
        S_{\ds,\ds'}\coloneqq\\
        \left\{ \ds_{\beta}~ \middle\vert~ \begin{array}{l}
    x(\ds_{\beta})_n\leq \lfloor \left(1-\exp\left(-\epsilon\right)\right)\cdot x_n\rfloor\\
    x(\ds_{\beta})_{n-1}\leq \lfloor \left(1-\exp\left(-\epsilon\right)\right)\cdot (x_{n-1}+1)\rfloor
  \end{array}\right\}\,.
   \end{multline*}
   For any sample dataset $\ds_{\beta} \in S_{\ds,\ds'}$,
   \begin{align*}
      & &x(\ds_{\beta})_n&\leq \lfloor \left(1-\exp\left(-\epsilon\right)\right)\cdot x_n\rfloor\,,\\
       &\implies &x(\ds_{\beta})_n&\leq \left(1-\exp\left(-\epsilon\right)\right)\cdot x_n\,,\\
      &\implies& \exp\left(-\epsilon\right)x_n&\leq x_n-x(\ds_{\beta})_n\\
       &\implies &\frac{x_n}{x_n-x(\ds_{\beta})_n}&\leq \exp\left(\epsilon\right)\,.
   \end{align*}
   Notice that the count $x(\ds_{\beta})_n$ is non-negative, which implies the following
   \begin{equation}\label{eq:beta_mech_priv_1}
       1\leq \frac{x_n}{x_n-x(\ds_{\beta})_n}\leq \exp\left(\epsilon\right)\,.
   \end{equation}
   Likewise, for the attribute $\bm{a}_{n-1}$, the following inequalities hold
   \begin{equation}\label{eq:beta_mech_priv_2}
       \exp\left(-\epsilon\right)\leq \frac{x_{n-1}+1-x(\ds_{\beta})_{n-1}}{x_{n-1}+1}\leq 1\,.
   \end{equation}
   Therefore,
   \begin{align*}
       &\frac{\pr{\bmech(\ds)=\ds_{\beta}}}{\pr{\bmech(\ds')=\ds_{\beta}}}\\
       =~&\frac{ \pr{x(\ds_{\beta})_{n-1}\mid x_{n-1}}}{\pr{x(\ds_{\beta})_{n-1}\mid x(\ds')_{n-1}}}\cdot 
       \frac{ \pr{x(\ds_{\beta})_{n}\mid x_{n}}}{\pr{x(\ds_{\beta})_{n}\mid x(\ds')_{n}}}\\
       =~&\frac{\binom{x_{n}}{x(\ds_{\beta})_{n}}\binom{x_{n-1}}{x(\ds_{\beta})_{n-1}}
       \left(\frac{\beta}{1-\beta}\right)^{\!x(\ds_{\beta})_{n}+x(\ds_{\beta})_{n-1}}\!\!\!\!(1-\beta)^{x_n+x_{n-1}}}{\binom{x_{n}-1}{x(\ds_{\beta})_{n}}\binom{x_{n-1}+1}{x(\ds_{\beta})_{n-1}}
       \left(\frac{\beta}{1-\beta}\right)^{\!x(\ds_{\beta})_{n}+x(\ds_{\beta})_{n-1}}\!\!\!\!(1-\beta)^{x_n+x_{n-1}}}\\
       =~&\frac{x_n}{x_n-x(\ds_{\beta})_n}\cdot  \frac{x_{n-1}+1-x(\ds_{\beta})_{n-1}}{x_{n-1}+1}\,.
   \end{align*}
   Then, by Equation \eqref{eq:beta_mech_priv_1} and \eqref{eq:beta_mech_priv_2},
   \begin{equation*}
       \exp\left(-\epsilon\right)\leq \frac{\pr{\cM_{\beta}(\ds)=\ds_{\beta}}}{\pr{\cM_{\beta}(\ds')=\ds_{\beta}}}\leq \exp\left(\epsilon\right)\,.
   \end{equation*}
   Thus, for any neighboring datasets $\ds$ and $\ds'$,
    $S_{\ds,\ds'}$ turns out to be a set, each element of which
    $\cM_{\beta}(\ds)$ and $\cM_{\beta}(\ds')$ generate with similar probability.
   By Lemma \ref{lem:dp_cond}, the error probability $\delta$ can then be
   computed as
   \begin{align*}
       \delta&=\max_{\ds\sim\ds'}~\pr{\cM_{\beta}(\ds)\in S_{\ds,\ds'}^{\complement}}\\
       &=1-\min_{\ds\sim\ds'}~\pr{\cM_{\beta}(\ds)\in S_{\ds,\ds'}}\\
       &=1-\min_{w\in [B]}~
       \left(\sum_{j=0}^{\nu}f(j;w,\beta)\right)^2\,,
   \end{align*}
   where $\nu$ is a shorthand for $\lfloor(1-\exp(-\epsilon))w\rfloor$ and
   $f$ is a probability mass function 
   defined in Equation \eqref{eq:binom_pmf}. 
   In this way, the privacy guarantee has been established for the
   DP $k$-anonymity mechanism.
\end{proof}

\begin{proof}[Proof of Theorem \ref{prop:sup_fair_bound}]
Observe that, for any $i<j$,
\begin{align}
    &\biasi{\tmech}{i} - \biasi{\tmech}{j}=\left(\frac{\thresh}{2}-x_i\right) p_i-\left(\frac{\thresh}{2}-x_j\right) p_j\nonumber\\
    =~&\left[\left(\frac{\thresh}{2}-x_i\right) p_i-\left(\frac{\thresh}{2}-x_j\right) p_i\right]+\nonumber\\
    &\left[\left(\frac{\thresh}{2}-x_j\right) p_i-\left(\frac{\thresh}{2}-x_j\right) p_j\right]\nonumber\\
    =~&\left(x_j-x_i\right)p_i+\left(\frac{\thresh}{2}-x_j\right)(p_i-p_j)\label{eq:thresh_mech_a_fair_1}
\end{align}
It follows that the fairness error $\acs$ can be computed as
    \begin{align}
        &\norm{\bias{\tmech}}_{\rightleftharpoons}=\max_{i\in [n]}~\biasi{\tmech}{i}-
        \min_{j\in[n]}~\biasi{\tmech}{j}\nonumber\\
        =~&\max_{1\leq i<j\leq n}~\left\vert \biasi{\tmech}{i} - \biasi{\tmech}{j}\right\vert\nonumber\\
        =~&\max_{1\leq i<j\leq n}~\left\vert \left(x_j-x_i\right)p_i+\left(\frac{\thresh}{2}-x_j\right)(p_i-p_j)\right\vert\label{eq:thresh_mech_a_fair_2}\\
        \leq~&\max_{1\leq i<j\leq n}~\left\vert \left(x_j-x_i\right)p_i\right\vert+\max_{1\leq i<j\leq n}~\left\vert \left(\frac{\thresh}{2}-x_j\right)(p_i-p_j)\right\vert\label{eq:thresh_mech_a_fair_3}\\
        =~&(x_n-x_1)p_1+\max_{1\leq i<j\leq n}~\left\vert \left(\frac{\thresh}{2}-x_j\right)(p_i-p_j)\right\vert\label{eq:thresh_mech_a_fair_4}\\
        \leq~&(x_n-x_1)p_1+(p_1-p_n)\max_{j\in[n]}~\left\vert \left(\frac{\thresh}{2}-x_j\right)\right\vert\nonumber\\
        \leq~&\left(x_n-x_1\right)p_1+\max\left\{\left\vert\frac{\thresh}{2}-x_1\right\vert,\left\vert\frac{\thresh}{2}-x_n\right\vert\right\}(p_1-p_n)\,,\nonumber
        \end{align}
        where Equation \eqref{eq:thresh_mech_a_fair_2} is derived from Equation \eqref{eq:thresh_mech_a_fair_1} and Equation \eqref{eq:thresh_mech_a_fair_3}
        comes from the triangle inequality. Besides, Equation \eqref{eq:thresh_mech_a_fair_4}
        is due to the fact that the histogram $\bm{x}$ is sorted in an increasing order, i.e.,
        $x_1\leq\dots\leq x_n$ and, as a consequence, 
        the probabilities $\bm{p}$ in Equation \eqref{eq:p_vector} appear in a decreasing
        order, i.e., $p_1\geq \dots\geq p_n$. Except for
        the trivial case that, the
        the counts are all the same, i.e., $x_1=\dots=x_n$,
        this inequality
        is also tight when the maximum count of the original 
        histogram is exactly half of the 
        threshold, i.e., $x_n=k/2$.
        
\end{proof}

\begin{proof}[Proof of Proposition \ref{prop:pram_bias_exp}]
    For any $\bm{a}_i\in\cX$,  it follows that
    \begin{align}
        &\biasi{\prammech}{i}=\EE{}{\prammech(\ds)_i}-x_i\nonumber\\
        =~&\sum_{j=1}^n\EE{}{x_j\cdot \bm{1}\left\{\tilde{\bm{a}}_j=\bm{a}_i\mid \bm{a}_i\right\}}-x_i\nonumber\\
        =~&\EE{}{x_i\cdot \bm{1}\left\{\tilde{\bm{a}}_i=\bm{a}_i\mid \bm{a}_i\right\}}+\nonumber\\
        &\sum_{j\in\cI_i\setminus\{i\}}\EE{}{x_j\cdot \bm{1}\left\{\tilde{\bm{a}}_j=\bm{a}_i\mid \bm{a}_i\right\}}+\nonumber\\
        &\sum_{j\in[n]\setminus\cI_i}\EE{}{x_j\cdot \bm{1}\left\{\tilde{\bm{a}}_j=\bm{a}_i\mid \bm{a}_i\right\}}
        -x_i\nonumber\\
        =~& x_i\cdot \gamma +\sum_{j\in\cI_i\setminus\{i\}}\left( x_j\cdot\frac{1-\gamma}{\qiattrs - 1}\right) + 0 -x_i\nonumber\\
        =~&\frac{x_i\exp(\epsilon)}{\exp(\epsilon)+n_Q-1} +\frac{\sum_{j\in\cI_i} x_j-x_i}{\exp(\epsilon)+n_Q-1}+0-x_i\label{eq:dp_swap_bias_1}\\
        =~&\frac{\sum_{j\in\cI_i} x_j - \qiattrs x_i}{\exp(\epsilon)+\qiattrs-1}\nonumber\,,
    \end{align}
    where Equation \eqref{eq:dp_swap_bias_1} just
    plugs in the value $\gamma$ defined in Equation \eqref{eq:dp_swap_def}.
    Thus, it manages to establish the mathematical expression of the bias of the DP swapping
    mechanism $\prammech$. 
\end{proof}

\begin{proof}[Proof of Theorem \ref{prop:pram_fair_bound}]
In the first place, notice that, for any $i\in[n]$,
\begin{equation*}
    \sum_{j\in\cI_i}\biasi{\prammech}{j}=\sum_{j\in\cI_i}
    \frac{\sum_{l\in\cI_i} x_l - \qiattrs x_j}{\exp(\epsilon)+\qiattrs-1}=0\,,
\end{equation*}
which implies the following
\begin{equation}\label{eq:pram_fair_1}
    \max_{j\in \cI_i}~\biasi{\prammech}{j}\geq 0\geq
    \min_{j\in \cI_i}~\biasi{\prammech}{j}\,.
\end{equation}
    Suppose that $\overline{g}$ and
    $\underline{h}$ are the indices associated with
    maximum and minimum biases respectively, i.e.,
    \begin{equation*}
        \overline{g} = \underset{l\in [n]}{\arg\max}~\biasi{\prammech}{l}\,,\quad
        \underline{h} = \underset{l\in [n]}{\arg\min}~\biasi{\prammech}{l}\,.
    \end{equation*}
    and $\underline{g}$ (or $\overline{h}$) represents
    the index associated with the minimum (or maximum) bias
    over the index set $\cI_{\overline{g}}$ (or $\cI_{\underline{h}}$), i.e.,
    \begin{equation*}
        \underline{g} = \underset{l\in \cI_{\overline{g}}}{\arg\min}~\biasi{\prammech}{l}\,,\quad
        \overline{h} = \underset{l\in \cI_{\underline{h}}}{\arg\max}~\biasi{\prammech}{l}\,.
    \end{equation*}
    By Equation \eqref{eq:pram_fair_1}, the following
    inequalities hold
 \begin{equation}\label{eq:pram_fair_2}
        \biasi{\prammech}{\underline{g}}\leq 0\leq \biasi{\prammech}{\overline{h}}\,.
    \end{equation}
    Then, it follows that
    \begin{align}
        &\norm{\bias{\prammech}}_{\rightleftharpoons}=
        \biasi{\prammech}{\overline{g}} - \biasi{\prammech}{\underline{h} }\nonumber\\
    \leq~& \left(\biasi{\prammech}{\overline{g}} - \biasi{\prammech}{\underline{g}}\right) +
    \left(\biasi{\prammech}{\overline{h}} - \biasi{\prammech}{\underline{h} }\right)\label{eq:pram_fair_3}\\
    =~&\frac{\qiattrs \left(x_{\underline{g}}-x_{\overline{g}}\right)}{\exp(\epsilon)+\qiattrs-1}+\frac{\qiattrs \left(x_{\underline{h}}-x_{\overline{h}}\right)}{\exp(\epsilon)+\qiattrs-1}\nonumber\\
    \leq ~&\frac{2\qiattrs \left(x_n-x_1\right)}{\exp(\epsilon)+\qiattrs-1}\nonumber\\
    =~&\frac{2\qiattrs \norm{\bm{x}}_{\rightleftharpoons}}{\exp(\epsilon)+\qiattrs-1}\nonumber\,,
    \end{align}
    where Equation \eqref{eq:pram_fair_3} is a direct
    consequence of Equation \eqref{eq:pram_fair_2}.
\end{proof}

\begin{proof}[Proof of Theorem \ref{prop:lap_fair_bound}]
    By Definition \ref{def:a-fair} of $\alpha$-fairness, the fairness violation coefficient $\alpha$ can be computed as
    \begin{align}
        \norm{\bias{\lapmech}}_{\rightleftharpoons}&=
        \max_{j\in [n]}~\biasi{\lapmech}{j}- \min_{j\in [n]}~\biasi{\lapmech}{j}\nonumber\\
        &=\biasi{\lapmech}{1} - \biasi{\lapmech}{n}\label{eq:lap_mech_fair_aux_1}\\
        &=\frac{\exp\left(-\epsilon x_1/2\right)-\exp\left(-\epsilon x_n/2\right)}{\epsilon} \nonumber\\
        &\leq \frac{x_n-x_1}{\epsilon} \sup_{x\in (x_1,x_n)}\left\vert\frac{d \exp\left(-\epsilon x/2\right)}{dx}\right\vert\label{eq:lap_mech_fair_aux_2}\\
        &=\frac{\exp\left(-\epsilon x_1/2\right)}{2}\left(x_n-x_1\right)\nonumber\\
        &=\frac{\exp\left(-\epsilon x_1/2\right)}{2}\norm{\bm{x}}_{\rightleftharpoons}\nonumber\,,
    \end{align}
    where Equation \eqref{eq:lap_mech_fair_aux_1} comes from the fact that the biases decrease,
    as the counts increase, i.e., $\biasi{\lapmech}{1}\geq\dots\geq \biasi{\lapmech}{n}\geq 0$.
    Besides, Equation \eqref{eq:lap_mech_fair_aux_2} is due to the mean value inequalities \citep{clarke1994mean}. It completes the proof here.
\end{proof}

\begin{proof}[Proof of Theorem \ref{thm:fair-comp}]
    First of all, note that
    \begin{align}
        \acs&=\left(x_n-x_1\right)p_1+\max\left\{\left\vert\frac{\thresh}{2}-x_1\right\vert,\left\vert\frac{\thresh}{2}-x_n\right\vert\right\}(p_1-p_n)\nonumber\\
        &\geq \left(x_n-x_1\right)p_1\nonumber\\
        &\geq \frac{1}{2} (x_n - x_1)\label{eq:fair-comp-1}\\
        &\geq \frac{\exp\left(-\epsilon x_1/2\right)}{2}\left(x_n-x_1\right)\nonumber\\
        &=\alap\nonumber\,,
    \end{align}
    where the inequality in Equation \eqref{eq:fair-comp-1} comes from
    the fact that the function $x\mapsto \pr{x+\eta\leq \thresh}$ with $\eta\sim\lap{\nicefrac{2}{\epsilon}}$ is 
    decreasing and $x_1$ is below the threshold $\thresh$, which implies the following
    \begin{equation*}
        p_1=\pr{x+\eta\leq \thresh}\geq \pr{\thresh+\eta\leq \thresh}=\frac{1}{2}\,.
    \end{equation*}
    Besides, $\qiattrs$ is the cardinality of the restricted data
    universe $\qiuniverse$, which is assumed to be non-empty
    and thus $\qiattrs$ is at least $1$.
    Then, it follows that
    \begin{align}
        \asw&=\frac{2\qiattrs}{\exp(\epsilon)+\qiattrs-1}(x_n-x_1)\nonumber\\
        &=2\left(1-\frac{\exp(\epsilon)-1}{\exp(\epsilon)+\qiattrs-1}\right)(x_n-x_1)\nonumber\\
        &\geq 2\exp\left(-\epsilon\right)(x_n-x_1)\nonumber\\
        &\geq 2\exp\left(-\epsilon x_1/2\right)(x_n-x_1)\label{eq:fair-comp-2}\\
        &\geq \frac{\exp\left(-\epsilon x_1/2\right)}{2}\left(x_n-x_1\right)\nonumber\\
        &=\alap\nonumber\,,
    \end{align}
    where Equation \eqref{eq:fair-comp-2} is based on the assumption that
    $x_1$ is no less than $2$.
\end{proof}

\section{Traditional DA Algorithms}
\label{app:sec:DA_alg}
This section presents more formal specifications of  the traditional DA algorithms adopted in the paper.

\subsection{DP Cell Suppression}
\label{app:dp_cell_suppression}

Algorithms \ref{das:cellsupp} and \ref{das:dpcellsupp} provide the pseudocode for the traditional cell suppression mechanism and its differentially private counterpart, respectively. 

To further elaborate, algorithm \ref{das:cellsupp}, which describes the traditional cell suppression mechanism, takes as input a histogram $\bm{x}(D)$ and a threshold $k\in\mathbb{Z}_+$ and returns a private version $\tilde{\bm{x}}(D)$ of $\bm{x}(D)$. The algorithm iterates through each record of the histogram and suppresses each value $x_i$ with value $\lfloor\nicefrac k2\rfloor$ if $x_i < k$ or releases the original value $x_i$ otherwise (lines 1--3).
\begin{algorithm}[h]
\caption{Cell Suppression}\label{das:cellsupp}
\begin{algorithmic}[1]
    \item[\algorithmicrequire] Histogram $\bm{x}(D)$ with vector of counts $(x_i)_{i\in[n]}$, threshold $k\in\mathbb{Z}_+$.
\item[\algorithmicfunction]{cellSuppress}{($\bm{x}(D),k$):}
\For{$i\in[n]$}

\State $\tilde x_i\gets\begin{cases}\lfloor\nicefrac{k}{2}\rfloor\text{ if } x_i<k\\x_i \text{ otherwise}\end{cases}$
\EndFor
\State\algorithmicreturn$\,\text{Histogram }\tilde{\bm{x}}(D)\text{ with counts } \tilde{\bm x}=(\tilde{x}_i)_{i\in[n]}$
\end{algorithmic}
\end{algorithm}

Algorithm \ref{das:dpcellsupp} describes the differentially private counterpart of cell suppression. It takes as input a histogram $\bm{x}(D)$, a threshold $k\in\mathbb{Z}_+$, and a privacy parameter $\eps>0$ and returns a private version $\tilde{\bm{x}}(D)$ of the original histogram $\bm{x}(D)$. First the threshold $k$ is perturbed with Laplace noise (of scale $\nicefrac{2}{\eps}$) to obtain $\tilde{k}$ (line 1). Then the algorithm iterates over every record of the histogram and suppresses each count $x_i$ with value $\lfloor\nicefrac{\tilde{k}}{2}\rfloor$ if $x_i<\tilde{k}$ or releases the original value $x_i$ otherwise (lines 2--4).

\begin{algorithm}[h]
\caption{DP Cell Suppression}\label{das:dpcellsupp}
\begin{algorithmic}[1]
\item[\algorithmicrequire] Histogram $\bm{x}(D)$ with vector of counts $(x_i)_{i\in[n]}$, threshold $k\in\mathbb{Z}_+$, privacy parameter $\eps>0$.
\item[\algorithmicfunction]{DPCellSuppress}{($\bm{x}(D),k,\eps$):}
\State $\tilde k\gets k+\text{Laplace}(\nicefrac{2}{\eps})$
\For{$i\in[n]$}
\State $\tilde x_i\gets\begin{cases}\lfloor \nicefrac{\tilde k}{2}\rfloor\text{ if } x_i<\tilde k\\x_i \text{ otherwise}\end{cases}$
\EndFor
\State\algorithmicreturn$\,\text{Histogram }\tilde{\bm{x}}(D)\text{ with counts } (\tilde x_i)_{i\in[n]}$
\end{algorithmic}
\end{algorithm}
\paragraph{On why cell suppression does not satisfy differential privacy.}
For instance, there exists a pair of neighboring datasets, $\ds$ and $\ds'$.
Suppose that $\ds$ has one more record of $\bm{a}_n$ than $\ds'$ while $\ds'$ has one more record of $\bm{a}_{n-1}$ than $\ds$.
The attributes $\bm{a}_{n-1}$ and $\bm{a}_n$ are assumed to be the ``majorities" in whichever dataset,
$\ds$ or $\ds'$, i.e.,
\begin{align*}
    \bm x(\ds')_{n-1} > \bm x(\ds)_{n-1}&\geq \thresh\,,\\
    \bm x(\ds)_{n} > \bm x(\ds')_{n}&\geq \thresh\,.
\end{align*}
Thus, given the input
datasets $\ds$ and $\ds'$,
the outputs of the original cell suppression mechanism associated with
$\bm{a}_n$ are still $\bm x(\ds)_n$ and $\bm x(\ds')_n = \bm x(\ds)_n-1$ respectively.
It means that this mechanism, due to the nature that it is deterministic, can hardly
derive the same output from these two neighboring datasets, which violates the requirements of differential privacy.
\subsection{DP Swapping}
\label{app:dp_swapping}
Swapping is done with respect to a metric that quantifies the discrepancies between any two records. Given the set of features $\Lambda$, this metric, let us denote it by $d_\text{swap}$, is defined over the domain of possible records of a histogram as
\begin{align*}
  d_\text{swap}(\bm{a}_i,\bm{a}_j)\triangleq
  &\sum_{\lambda\in\Lambda} \mathbbm{1}_\text{cat}(\lambda)\rho(\bm{a}_i[\lambda],\bm{a}_j[\lambda]) \\ 
  &+ \mathbbm{1}_\text{num}(\lambda)\frac{\vert \bm{a}_i[\lambda]-\bm{a}_j[\lambda]\vert}{\lambda_\text{range}}
\end{align*}
Where $\rho$ is the discrete metric (i.e. $\rho(a,b)=0\iff a=b$ and $1$ otherwise) and $\lambda_\text{range}$ is the range of the possible values taken by a numerical feature $\lambda$. 
$\mathbbm{1}_\text{cat}$ and $\mathbbm{1}_\text{num}$ are characteristic functions of the sets of categorical and numerical features of the histogram respectively. Refer to algorithm \ref{das:swapping} for details on the non-private/deterministic swapping algorithm.

Algorithm \ref{das:swapping} describes the traditional swapping mechanism. This takes as input a histogram $\bm{x}(D)$ with $N$ records, a swapping parameter $\gamma\in[0,1]$, and a list of quasi-identifiers. For $\lfloor\nicefrac{(1-\gamma)N}2\rfloor$ times, the algorithm picks a hitherto unswapped record $\bm{a}_i$ of the histogram, picks the closest unswapped record $\bm{a}_s$ to $\bm{a}_i$ (w.r.t. the metric $d_\text{swap}$) and swaps the quasi identifiers of $\bm{a}_i$ and $\bm{a}_s$ (lines 1--4).
\begin{algorithm}[!h]
\caption{Swapping}\label{das:swapping}
\begin{algorithmic}[1]
\item[\algorithmicrequire] Histogram $\bm{x}(D)$ of size $N$, Swapping Parameter $\gamma\in[0,1]$, list of quasi-identifiers $Q$
\item[\algorithmicfunction]{swapping}{($\bm{x}(D),\gamma,Q$):}
\For{$\lfloor\nicefrac{(1-\gamma)N}2\rfloor$ times}
\State Randomly pick an unswapped row $\bm{a}_i$ of $\bm{x}(D)$
\State $\bm{a}_s\gets\displaystyle{\arg\min_{\substack{\bm{a}_j\in \bm{x}(D)\setminus \bm{a}_i\\ \bm{a}_j\text{ is unswapped}}}}d_\text{swap}(\bm{a}_j,\bm{a}_i)$
\State $\bm{a}_i[Q],\bm{a}_s[Q]\gets \bm{a}_s[Q], \bm{a}_i[Q]$
\EndFor
\State\algorithmicreturn$\,$Swapped histogram $\bm{x}(D)$.
\end{algorithmic}
\end{algorithm}

The differentially private counterpart of swapping was described in subsection \ref{subsec:dpswap}. Algorithm \ref{das:dpswapping} presents this form of swapping. It takes as input a histogram $\bm{x}(D)$ with $N$ records, a privacy parameter $\eps>0$, and a list of quasi-identifiers $Q$ and returns a private/modified histogram $\bm{x}(D)$. For each record $\bm{a}_i$ of the histogram, the algorithm preserves it with probability $\gamma\triangleq\frac{\exp(\epsilon)}{\exp(\epsilon)+n_Q-1}$; else with probability $1-\gamma$ picks a set of values of quasi-identifiers from $\chi_Q\setminus \bm{a}_i[Q]$ uniformly at random, where $\chi_Q$ is the data universe of quasi-identifiers, and assigns it to $\bm{a}_i[Q]$ (lines 1--3).

\begin{algorithm}[!h]
    \caption{DP Swapping}\label{das:dpswapping}
    \begin{algorithmic}[1]
        \item[\algorithmicrequire] Histogram $\bm{x}(D)$ of size $N$, privacy parameter $\eps>0$, list of quasi-identifiers $Q$
        \item[\algorithmicfunction]{DPSwapping}{$(\bm{x}(D),\eps,Q):$}
        \For{row $\bm{a}_i$ in $D$}
        \State $\bm{a}_i[Q] \gets
\begin{cases}
  \bm{a}_i[Q] &\text{w.p. } \frac{\exp(\epsilon)}{\exp(\epsilon)+n_Q-1}, \\
\text{Uniform}(\qiuniverse\setminus \bm{a}_i[Q]) &\text{otherwise}
\end{cases}$
\EndFor
\State\algorithmicreturn$\,$Swapped histogram $\bm{x}(D)$ with rows $\{\bm{a}_i\}$
    \end{algorithmic}
\end{algorithm}
\subsection{DP $k$-anonymity}
\label{app:dp_kanonimity}
In this paper, to $k$-anonymize a dataset we utilize the Mondrian algorithm (\cite{LeFevre2006Mondrian}). This is a top-down greedy algorithm that takes a dataset as input and outputs a $k$-anonymized version of it. Interested readers may refer to the cited paper for details about this algorithm. \texttt{anonypy}, an anonymization package for python, includes an implementation for $k$-anonymity via the Mondrian algorithm, which has been used for the results on $k$-anonymity in this paper. 

\begin{algorithm}[!h]
    \caption{Producing Synthetic $k$-Anonymized Dataset}\label{das:kanon}
    \begin{algorithmic}[1]
        \item[\algorithmicrequire] Dataset $D$, anonymization parameter $k\in\mathbb{Z}_+$
        \item[\algorithmicfunction]{produceKanonymizedDataset}{$(D,k):$} 
        \State $k$-anonymize $D$ to get $D_\text{$k$-anon}$ using the Mondrian method (\cite{LeFevre2006Mondrian}). 
        \State$\tilde{D}\gets$reconstructDataset$(D_\text{$k$-anon})$
        \item[\algorithmicreturn] Reconstructed dataset $\tilde{D}$.
        \item[\algorithmicfunction] {reconstructDataset}{$(D_\text{$k$-anon},D):$}
        \State Initialise an empty dataset $\tilde{D}$ with the same set of features as $D$
        \For{every row $r$ in $D_\text{$k$-anon}$}
        \For{$r$[count] many times}
        \State Create new row $\tilde{r}$ for $\tilde{D}$
        \For{each feature $\lambda$}
        \If{$\lambda$ is categorical}
        \State Assign $\tilde{r}(\lambda)$ a value from the list $r[\lambda]$ uniformly at random.
        \Else
        \State Assign $\tilde{r}(\lambda)$ a value from the Gaussian $\mathcal{N}(\mu,\sigma)$ (rounded off to the nearest non-negative integer), where $\mu\triangleq\frac{a+b}2$ and $\sigma\triangleq\frac{b-a}{4}$, where $a, b$ are the endpoints of the interval $r[\lambda]\triangleq[a,b]$.
        \EndIf
        \EndFor
        \EndFor
        \EndFor
        \item[\algorithmicreturn] Reconstructed dataset $\tilde{D}$.
    \end{algorithmic}
\end{algorithm}

\begin{algorithm}[!h]
    \caption{DP $k$-Anonymity}\label{das:dpkanon}
    \begin{algorithmic}[1]
        \item[\algorithmicrequire] Dataset $D$, anonymization parameter $k\in\mathbb{Z}_+$, privacy parameter $\eps>0$
        \item[\algorithmicfunction]{produceDPKanonymizedDataset}{$(D,k,\eps):$} 
        \State $\beta\gets 1-(\exp(-\eps))$
        \State Create a subset $D'$ of the dataset $D$ by sampling from the rows of $D$ uniformly at random with probability $\beta$.
        \State $\tilde{D}\gets$produceKanonymizedDataset($D',k$)\algorithmiccomment{Using produceKanonymizedDataset from algorithm \ref{das:kanon}}
        \item[\algorithmicreturn] Reconstructed dataset $\tilde{D}$.
    \end{algorithmic}
\end{algorithm}

In the $k$-anonymized version, categorical attribute values are grouped together as lists and numerical ones are grouped together as intervals and each row is assigned a count attribute corresponding to how many rows of the original dataset the said row in the $k$-anonymized version represents. 

\paragraph{Reconstruction step.}
However, this makes it difficult to analyze the anonymized output with the original dataset in the same space. Thus it is necessary to reconstruct a synthetic dataset from the $k$-anonymized version that is in the same space as that of the original dataset. In our experiments, we include a reconstruction step for $k$-anonymity.

Algorithm \ref{das:kanon} describes how we obtain a reconstructed, privatized version of the dataset using the traditional $k$-anonymity algorithm. This algorithm involves two components: $k$-anonymization and reconstructing an output, privatized dataset in the space of the original dataset. It takes as input a dataset $D$ and $k\in\mathbb{Z}_+$ and outputs a reconstructed dataset $\tilde{D}$. 

First, the original dataset $D$ is $k$-anonymized using the Mondrian method (\cite{LeFevre2006Mondrian}) to obtain $D_\text{$k$-anon}$ (line 1). As $D_\text{$k$-anon}$ is not in the same space as $D$, the algorithm uses a reconstruction step (line 2). 

To perform the reconstruction, $D_\text{$k$-anon}$ is taken and a new empty dataset $\tilde{D}$ in the space of $D$ is created (line 3). The algorithm iterates over every row $r$ in the $k$-anonymized dataset for $r[\text{count}]$ times (i.e. once for every row in the original dataset that is represented by $r$ in $D_\text{$k$-anon}$), creates a new row $\tilde{r}$ for $D$; for each feature $\lambda$, if $\lambda$ is categorical, then the algorithm chooses one of the merged values of $\lambda$ in $r[\lambda]$ uniformly at random for $\tilde{r}[\lambda]$, or if $\lambda$ is numeric, then a random value is chosen from $\mathcal{N}(\mu,\sigma)$ and rounded off to the nearest non-negative integer, where $\mu$ is the midpoint of the interval $r[\lambda]$ and $\sigma$ is $\nicefrac{1}{4}$ times the length of the interval $r[\lambda]$ (lines 4-15). 

Algorithm \ref{das:dpkanon} describes the DP counterpart of the aforementioned $k$-anonymity process. It takes as input a dataset $D$, $k\in\mathbb{Z}_+$, and a privacy parameter $\eps>0$ and outputs a reconstructed, $k$-anonymized dataset $\tilde{D}$. The algorithm computes a sampling probability $\beta\triangleq 1-(\exp(-\eps))$ and samples rows of $D$ uniformly at random with probability $\beta$ to obtain $D'$ (lines 1--2). Then $D'$ and $k$ are passed as input to algorithm \ref{das:kanon} to produce the reconstructed dataset $\tilde{D}$ (line 3).
\section{Extended Results}
\label{app:results}

\subsection{Datasets and Settings adopted}
\label{app:datasets}
\paragraph{Datasets} The data adopted in our experimental studies was the Diverse Community Excerpts Benchmark Data, provided by the National Institute of Standards and Technology and available on the SDNist synthetic data evaluation library on GitHub.  The excerpts are a curated selection of geography and features derived from the American Community Survey (ACS). Each of these datasets is further divided into geographical regions known as PUMAs (Public Use Microdata Areas). In particular, we use data provided for Massachusetts, Texas, and Outlier PUMAs; these three datasets contain information about 5, 6, and 20 PUMAs respectively. 

In particular, while the full ACS data has about 200 features, the Diverse Community Excerpts Benchmark Data uses about 20 features. Out of these features, we use a slice of the dataset for our experiments corresponding to the features \texttt{['RACE', 'SEX', 'OWNERSHP', 
'AGE','INCTOT']}, which correspond to the race, sex, house ownership status, age and the total annual income of an individual respectively. Wherever necessary, the numeric features \texttt{AGE} and \texttt{INCTOT} are discretized/binarized respectively into groups. For instance, unless stated otherwise, we binarize \texttt{INCTOT} into whether a person earns more than $\$\, 50000$ per annum (1) or not (0). The age attribute, wherever used, may be discretized into groups/age brackets, depending upon the experiment. For example, our classification experiments involve dividing ages equitably into 5 age brackets. 

Note that in doing so, all features in the dataset slice being considered are now categorical, and this is especially convenient when it comes to the reconstruction step of the $k$-anonymity algorithm: now all the rows' features can be chosen uniformly at random from the list of merged attributes in the anonymized version rather than sampling from a Gaussian centered around the midpoint of an interval, which carries a slight risk of sampling values outside of the region defined by the endpoints of the interval.

\paragraph{Settings} These experiments have been coded and run using Python 3.9 and above. Some tasks involving heavy computation were performed using a cluster equipped with AMD EPYC 7452 32-Core CPUs (@ 1.5 GHz) and 8GB of RAM.
\subsection{Data Release}
\label{app:data_release}
Here empirical results on the data release via different mechanisms are provided for the Texas and Outlier datasets as done earlier for the Massachusetts dataset. 

Tables \ref{tab:TXcompareDPmethods} and \ref{tab:OTcompareDPmethods} provide the values of $\delta$, biases (w.r.t. the $\ell_1$ norm), and the fairness violation bound $\alpha$ respectively (wherever applicable, the threshold $k$ is set to be 6). 

Here, a similar trend is seen as for table \ref{tab:compareDPmethods} in the main text and the DP mechanisms (Laplace and discrete Gaussian) almost always offer lower values of $\delta$, biases, and $\alpha$ than the rest. This again demonstrates that DP methods do indeed offer better privacy protection, higher accuracy of data release, and better fairness guarantees than the other traditional DA mechanisms.

Figures \ref{fig:TX_error} and \ref{fig:OT_error} provide plots showing the errors ($\Vert\tilde{\bm x}-\bm x\Vert_1$) associated with the data release of each DA method and its DP variants. Figures \ref{fig:TX_fairness} and \ref{fig:OT_fairness} provide plots showing the fairness values ($\alpha$) associated with the same. 

As for the Massachusetts dataset in the main text, it is again seen here for the Texas and Outlier datasets that as $\eps$ increases, the DP counterpart of each DA mechanism approaches the original DA mechanism in terms of errors and fairness violations. This further reinforces the observation that these differentially private mechanisms are conceptually similar and perform similarly.

\begin{table}[!h]
\small
\centering
\resizebox{0.95\linewidth}{!}
{
    \begin{tabular}{c|l|r|r|r}
    \toprule
    $\epsilon$ & Mechanism & \multicolumn{1}{c|}{$\delta$} & Bias ($\ell_1$ norm) & $\alpha$-fairness \\
    \midrule
    \multirow{5}[2]{*}{0.5} & Laplace & \textbf{0} & \textbf{901.21
} & \textbf{3.645
} \\
          & Discrete Gaussian & 0.363  & 1156.63

  & 4.62

  \\
          & DP Suppression & 0.999  & 1138.62
  &  4.53
     \\
          & DP Swapping & 0.868 &  12988.58

 & 437.105
  \\
          & DP $k$-anonymity & 0.878  & 2963.3  & 24.7  \\
    \midrule
    \multirow{5}[2]{*}{1} & Laplace & \textbf{0} & \textbf{409.03
} & \textbf{1.815
} \\
          & Discrete Gaussian & 0.132  & 777.455
  & 2.96
      \\
          & DP Suppression & 0.999  & 1205.315
  & 4.47
  \\
          & DP Swapping & 0.874 &  11624.64
 &  394.28
 \\
          & DP $k$-anonymity & 0.906  & 4296.8  & 35.4 \\
    \midrule
    \multirow{5}[2]{*}{2} & Laplace & \textbf{0} & \textbf{187.59
    } & \textbf{0.905} \\
          & Discrete Gaussian & 0.017  & 523.825
  & 1.995
      \\
          & DP Suppression & 0.999  &  1219.19
 & 4.71
     \\
          & DP Swapping & 0.899 & 8212.7
  &  266.78
 \\
          & DP $k$-anonymity & 0.981  & 5406.7  &  43.4 \\
    \midrule
    \multirow{5}[2]{*}{4} & Laplace & \textbf{0} & \textbf{81.63
} & \textbf{0.46
} \\
          & Discrete Gaussian & 3E-4  & 353.99
  & 1.545
  \\
          & DP Suppression & 0.999  & 1217.115
  & 4.91
  \\
          & DP Swapping & 0.969 & 2117.62
  & 75.48
  \\
          & DP $k$-anonymity & 0.999  & 5992.0  & 48.9 \\
    \bottomrule
    \end{tabular}%
}
  \caption{TX dataset data release: Comparison of DP mechanisms in terms of $\delta$, $\ell_1$ norm of the empirical bias and $\alpha$-fairness. \label{tab:TXcompareDPmethods}}
\end{table}

\begin{table}[!h]
\small
\centering
\resizebox{0.95\linewidth}{!}
{
    \begin{tabular}{c|l|r|r|r}
    \toprule
    $\epsilon$ & Mechanism & \multicolumn{1}{c|}{$\delta$} & Bias ($\ell_1$ norm) & $\alpha$-fairness \\
    \midrule
    \multirow{5}[2]{*}{0.5} & Laplace & \textbf{0} & \textbf{2992.88    
} & \textbf{4.02
} \\
          & Discrete Gaussian & 0.363  & 3798.96
  & 4.885

  \\
          & DP Suppression & 0.999  & 3687.445
  &  4.39
     \\
          & DP Swapping & 0.868 &  35385.63
 & 580.015
  \\
          & DP $k$-anonymity & 0.878  & 10260.6   & 34.5 \\
    \midrule
    \multirow{5}[2]{*}{1} & Laplace & \textbf{0} & \textbf{1372.36} & \textbf{2} \\
          & Discrete Gaussian & 0.132  & 2570.38
  & 3.235

      \\
          & DP Suppression & 0.999  & 3911.82
  & 4.56
  \\
          & DP Swapping & 0.874 &  32134.71
 &  553.66  
 \\
          & DP $k$-anonymity & 0.906  & 14668.5  & 46.1 \\
    \midrule
    \multirow{5}[2]{*}{2} & Laplace & \textbf{0} & \textbf{628.335} & \textbf{0.975} \\
          & Discrete Gaussian & 0.017  & 1728.1

  & 2.62

      \\
          & DP Suppression & 0.999  &  3969.165
 & 4.76
     \\
          & DP Swapping & 0.899 & 22242.12
  &  339.09
 \\
          & DP $k$-anonymity & 0.981  & 18529.7  & 54.4  \\
    \midrule
    \multirow{5}[2]{*}{4} & Laplace & \textbf{0} & \textbf{274.36} & \textbf{0.48
} \\
          & Discrete Gaussian & 3E-4  & 1162

  & 1.71

  \\
          & DP Suppression & 0.999  & 3962.755
  & 4.93
  \\
          & DP Swapping & 0.969 & 5634.95
  & 108.565
  \\
          & DP $k$-anonymity & 0.999  & 20456.9  &   61.5   \\
    \bottomrule
    \end{tabular}%
}
  \caption{Outlier dataset data release: Comparison of DP mechanisms in terms of $\delta$, $\ell_1$ norm of the empirical bias and $\alpha$-fairness.\label{tab:OTcompareDPmethods}}
\end{table}

\subsection{Classification}
\label{app:classification}
\begin{figure*}[!t]
\centering
        \includegraphics[width=0.44\linewidth]{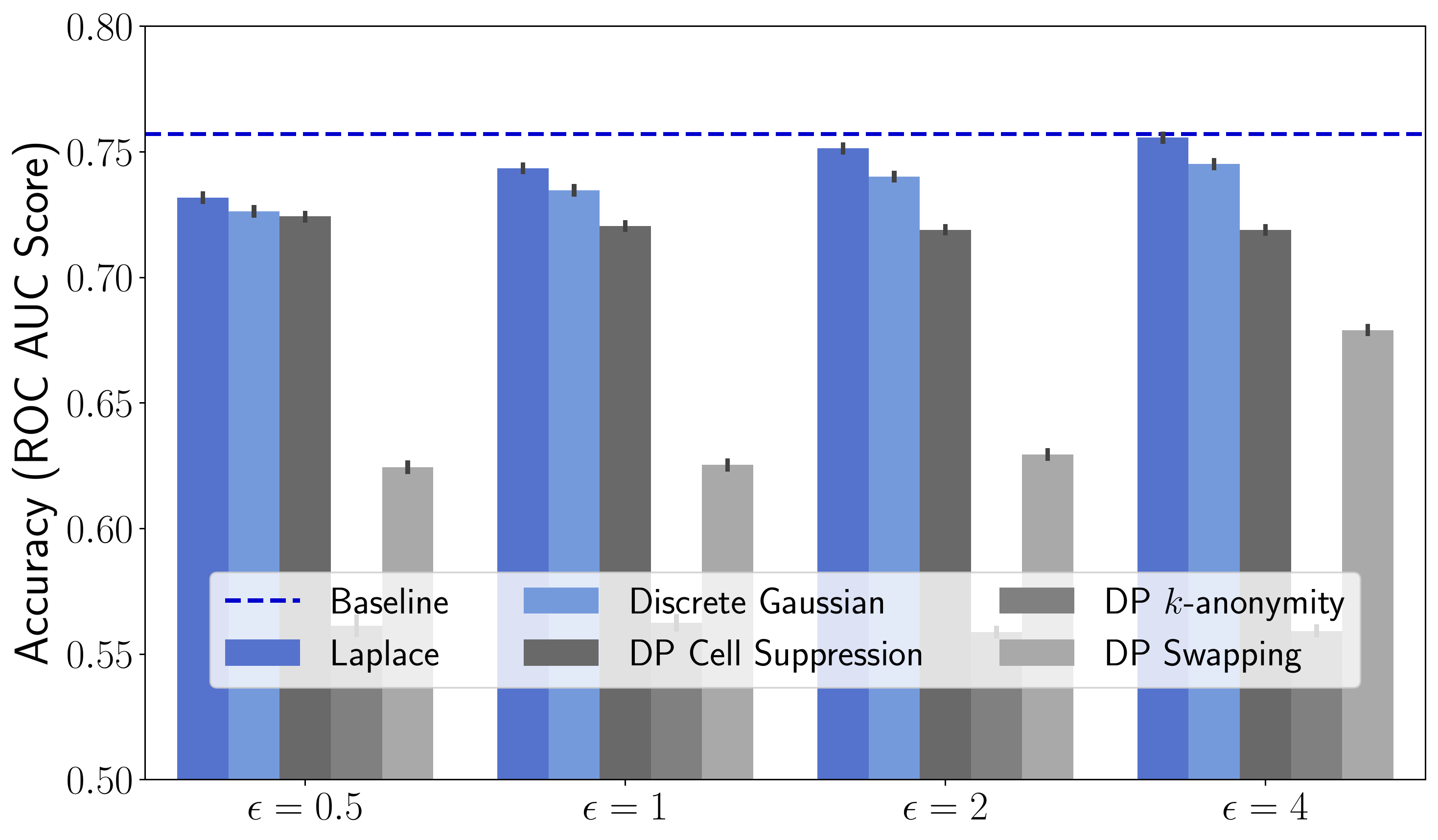}
        \includegraphics[width=0.44\linewidth]{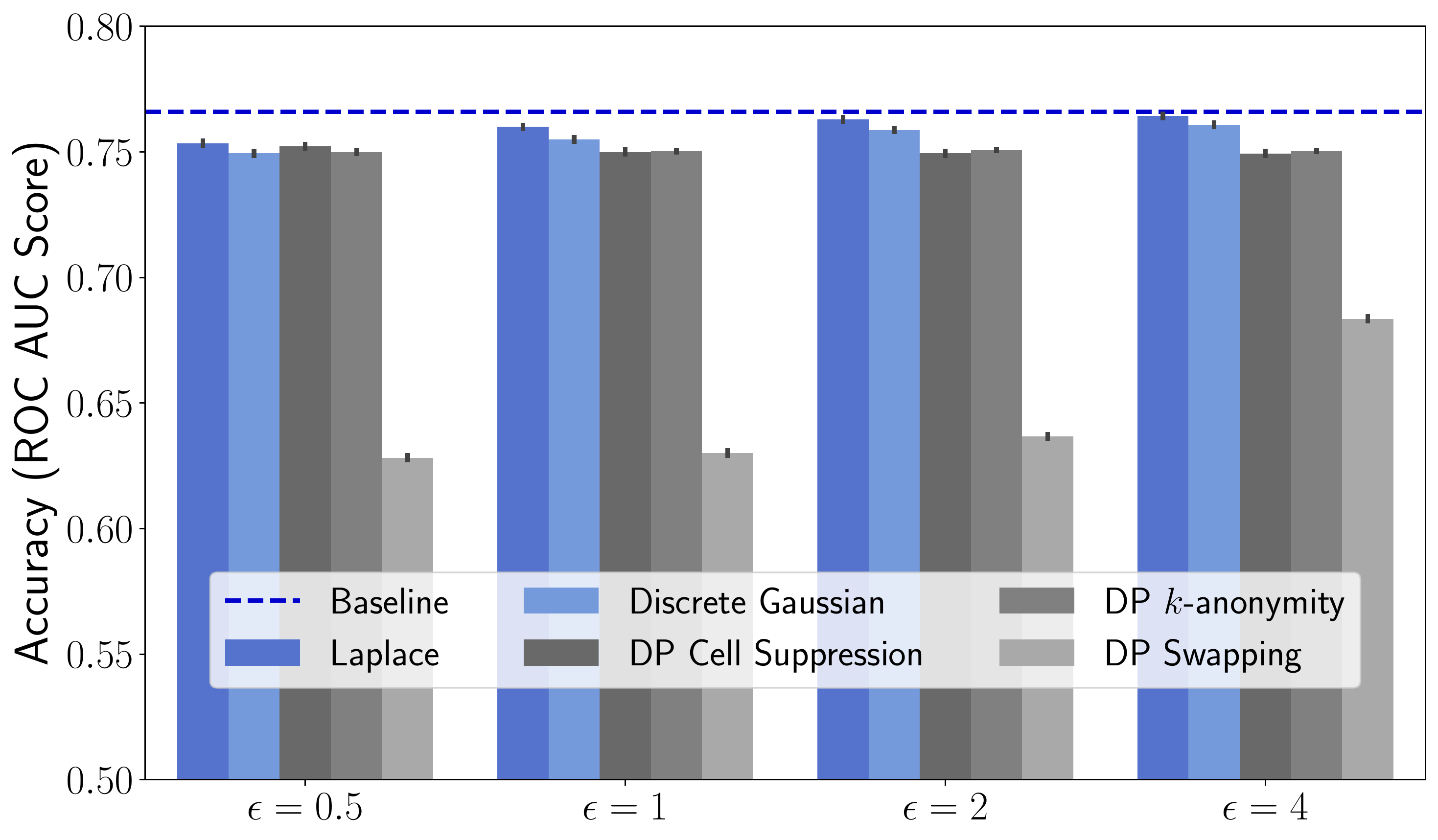}
    \caption{Results for Logistic Regression for Texas (left) and Outlier (right) datasets}
    \label{fig:OTLog_errors}
\end{figure*}
This subsection provides plots for accuracies of the logistic regression task described in the paper using various DA methods over the Texas and Outlier dataset
(Figure \ref{fig:OTLog_errors}). For these datasets as well, it is seen that training logistic regression classifiers with data produced by DP mechanisms (Laplace and discrete Gaussian) yields close-to-baseline classification accuracy. Also, it is seen that using data produced by traditional DA mechanisms yields accuracies that are lower and further away from the baseline accuracy than for any of the DP mechanisms.

\begin{figure*}[!t]
\centering
\includegraphics[width=0.3\textwidth]{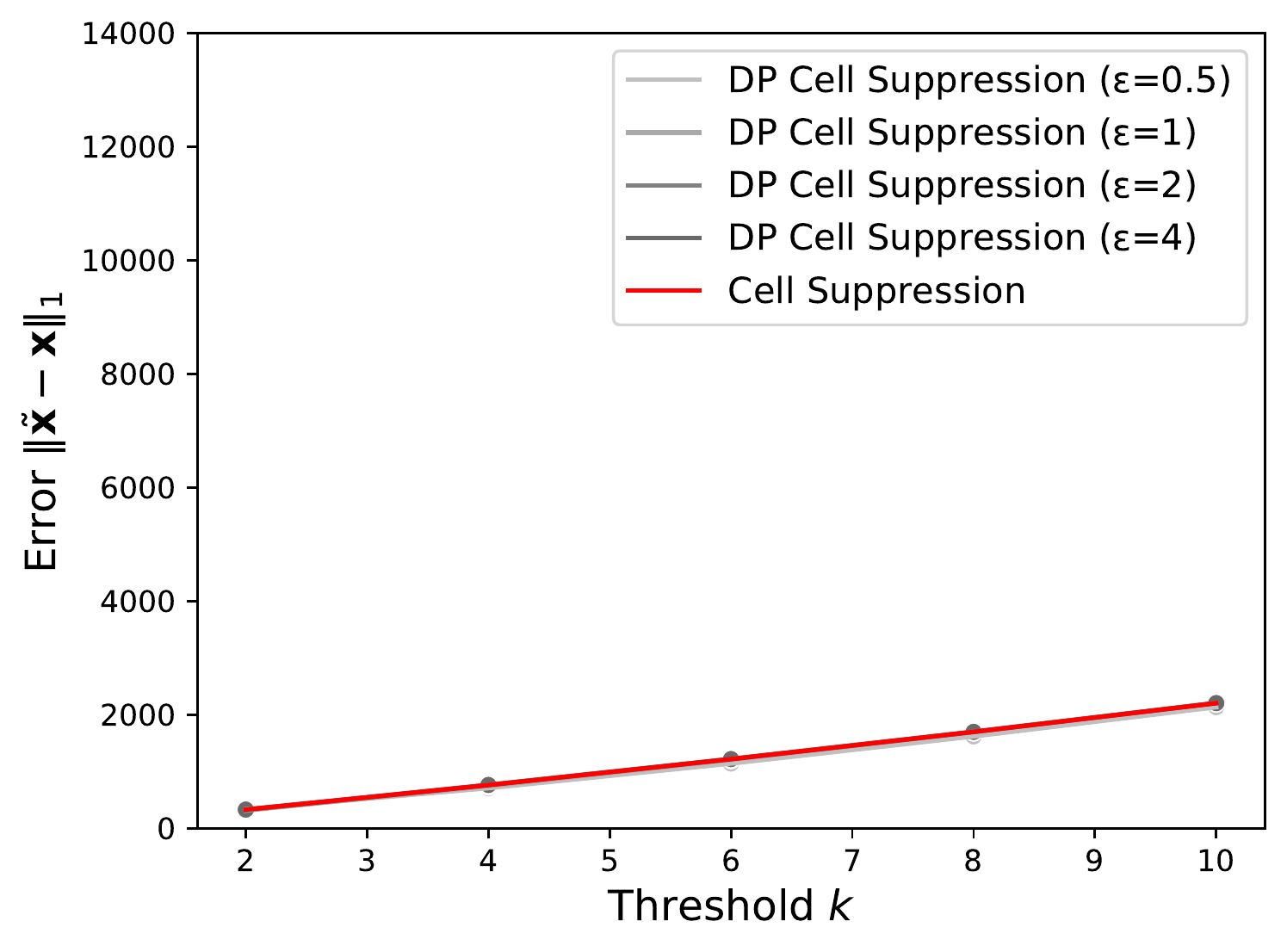}
\includegraphics[width=0.332 \textwidth]{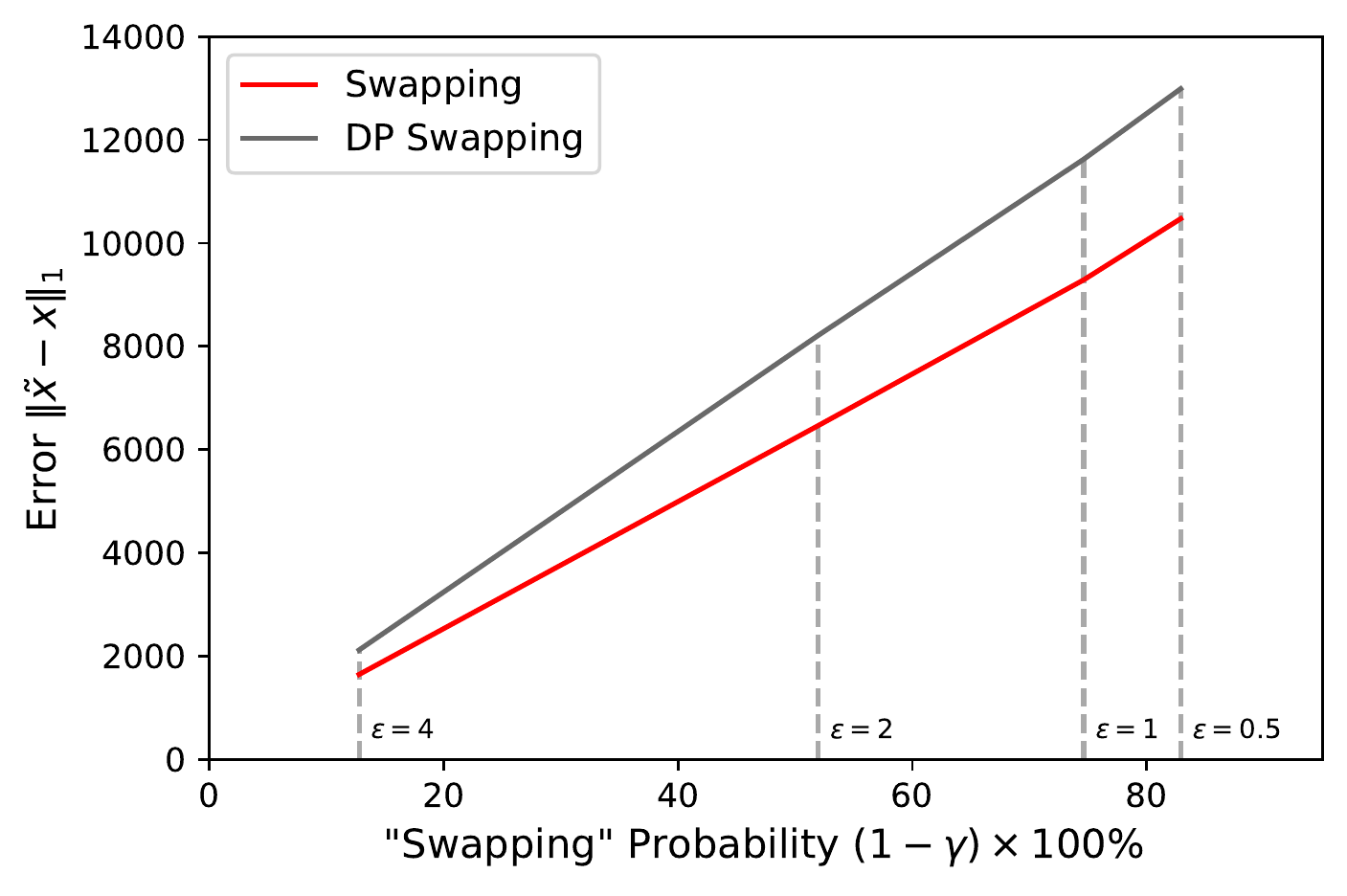}
\includegraphics[width=0.327\textwidth]{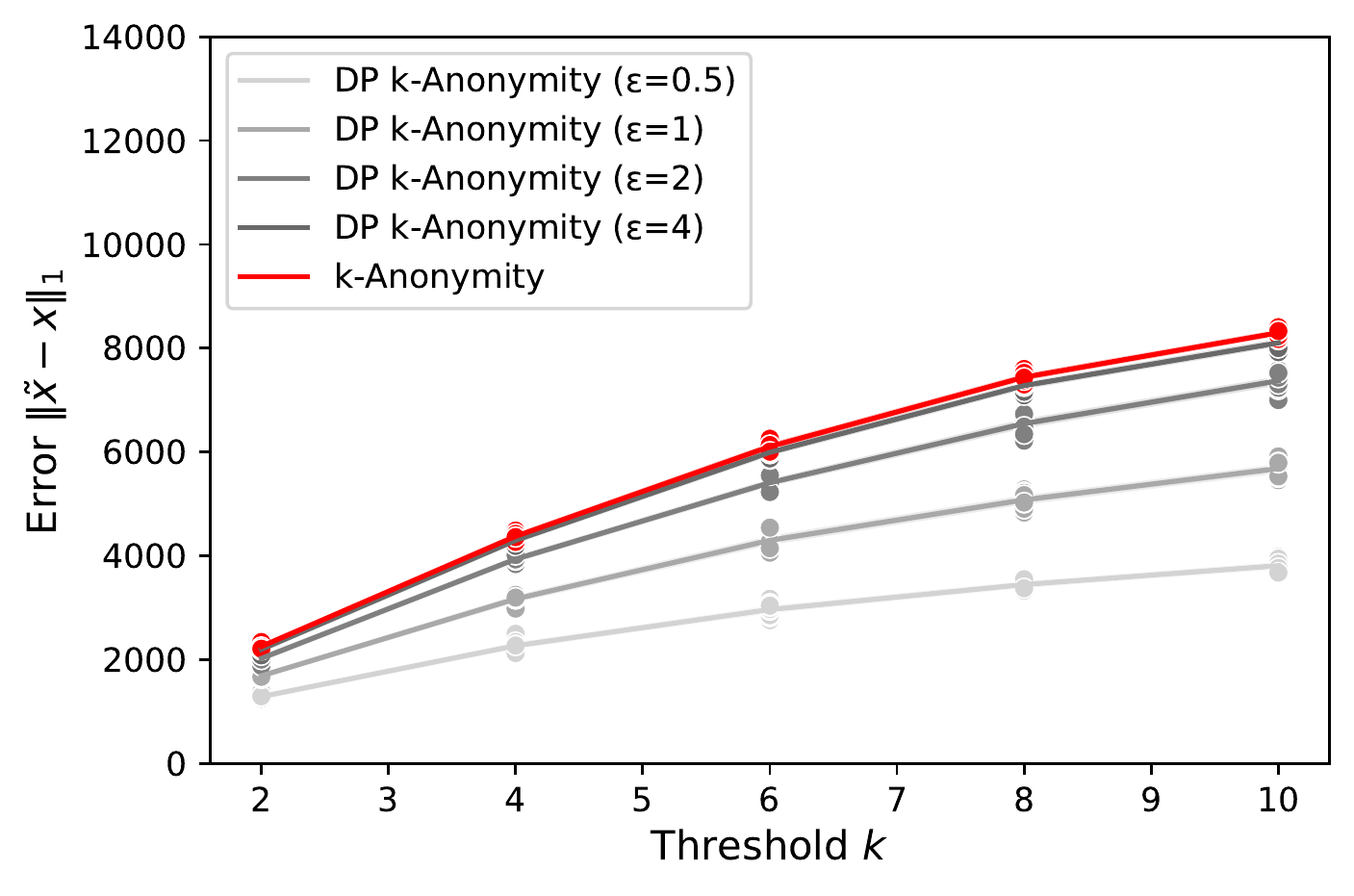}
\caption{TX ACS dataset: Errors $\Vert \tilde{\bm{x}}- \bm{x}\Vert_1$ for cell suppression (left), 
swapping (center) and $k$-anonymity (right) and their differentially 
private counterparts (average of 200 repetitions).}
\label{fig:TX_error}
\end{figure*} 

\begin{figure*}[!t]
\centering
\includegraphics[width=0.3\textwidth]{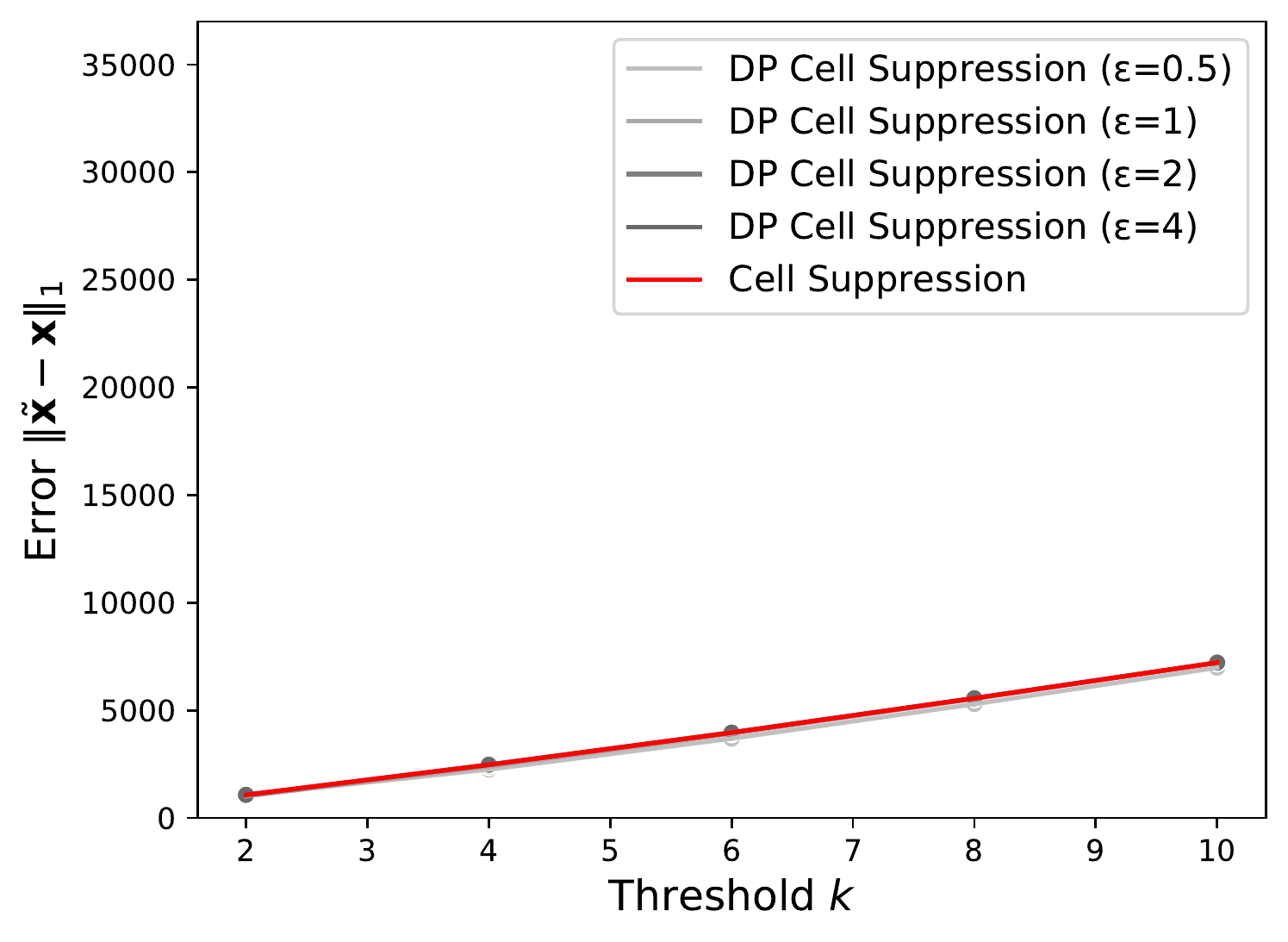}
\includegraphics[width=0.332 \textwidth]{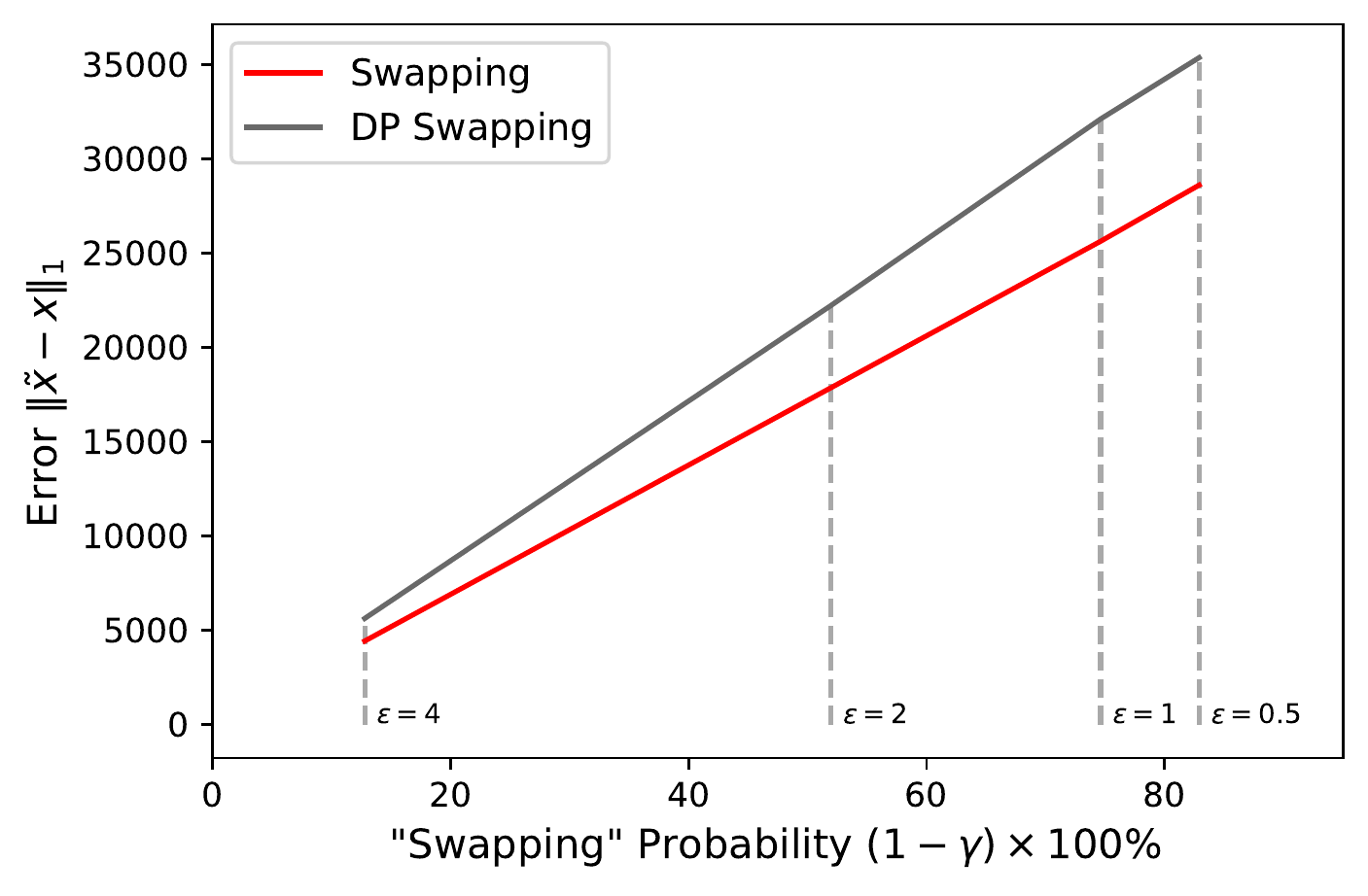}
\includegraphics[width=0.327\textwidth]{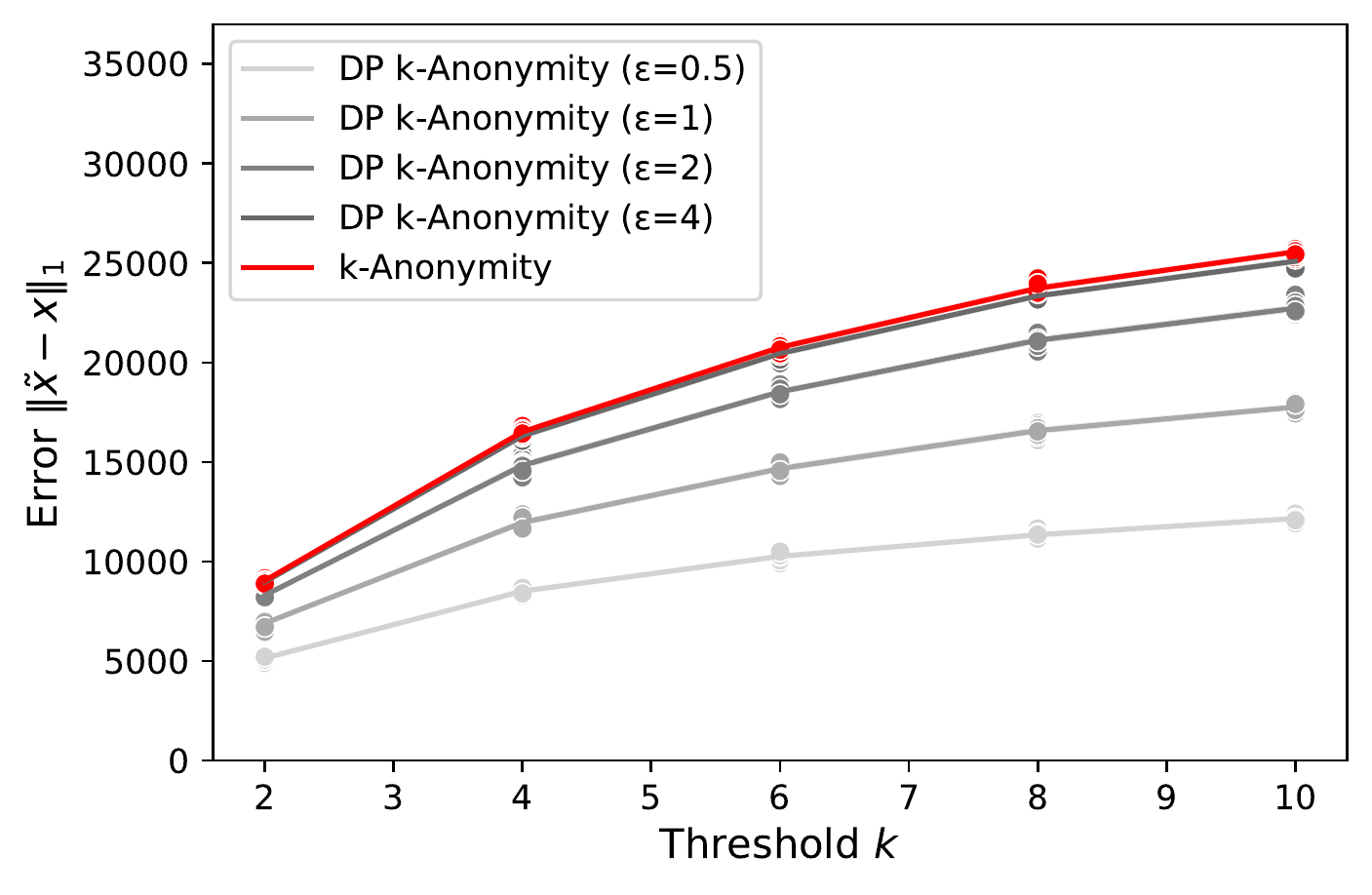}
\caption{Outlier ACS dataset: Errors $\Vert \tilde{\bm{x}}- \bm{x}\Vert_1$ for cell suppression (left), 
swapping (center) and $k$-anonymity (right) and their differentially 
private counterparts (average of 200 repetitions).}
\label{fig:OT_error}
\end{figure*}

\begin{figure*}[!t]
\centering
\includegraphics[width=0.3\textwidth]{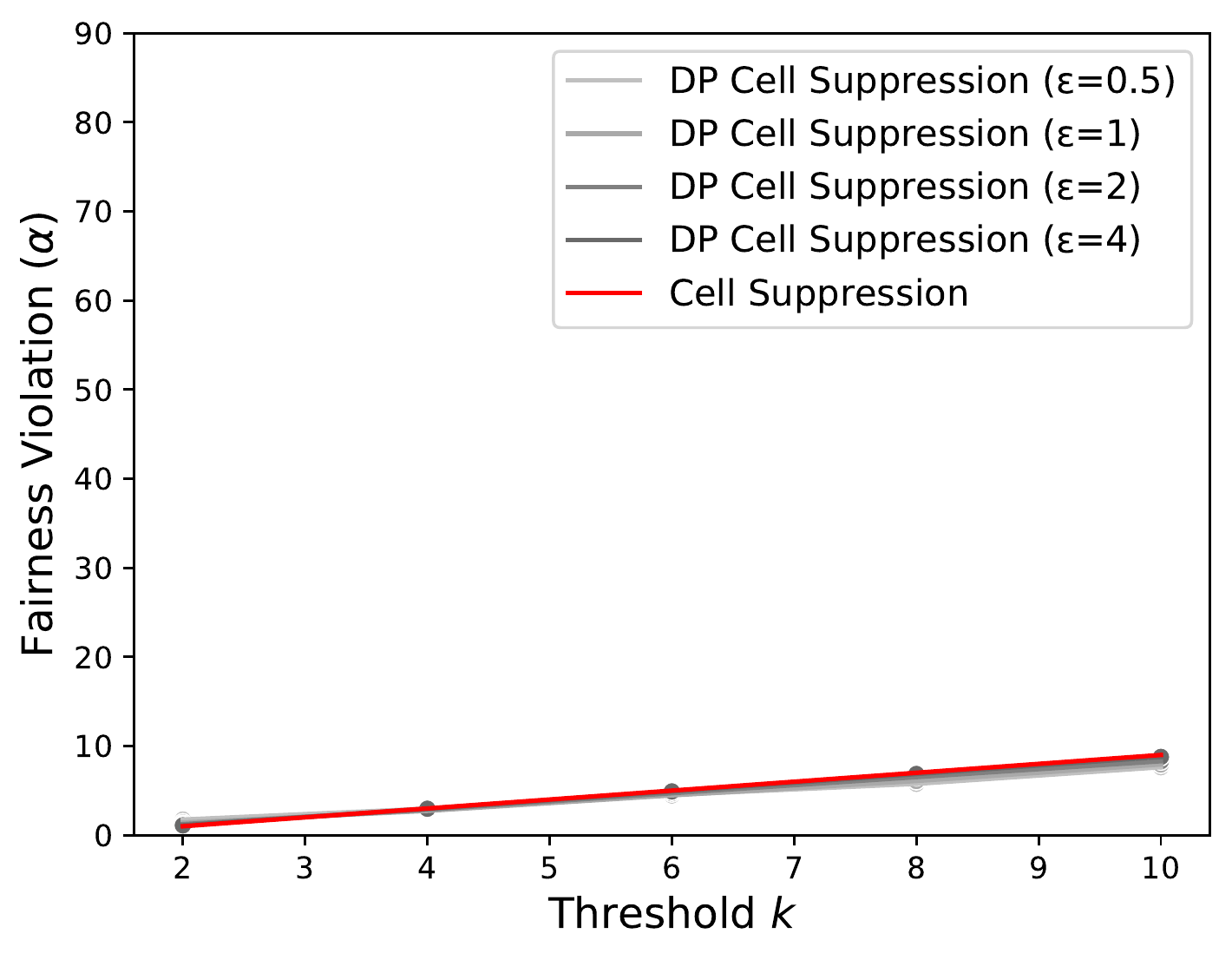}
\includegraphics[width=0.332 \textwidth]{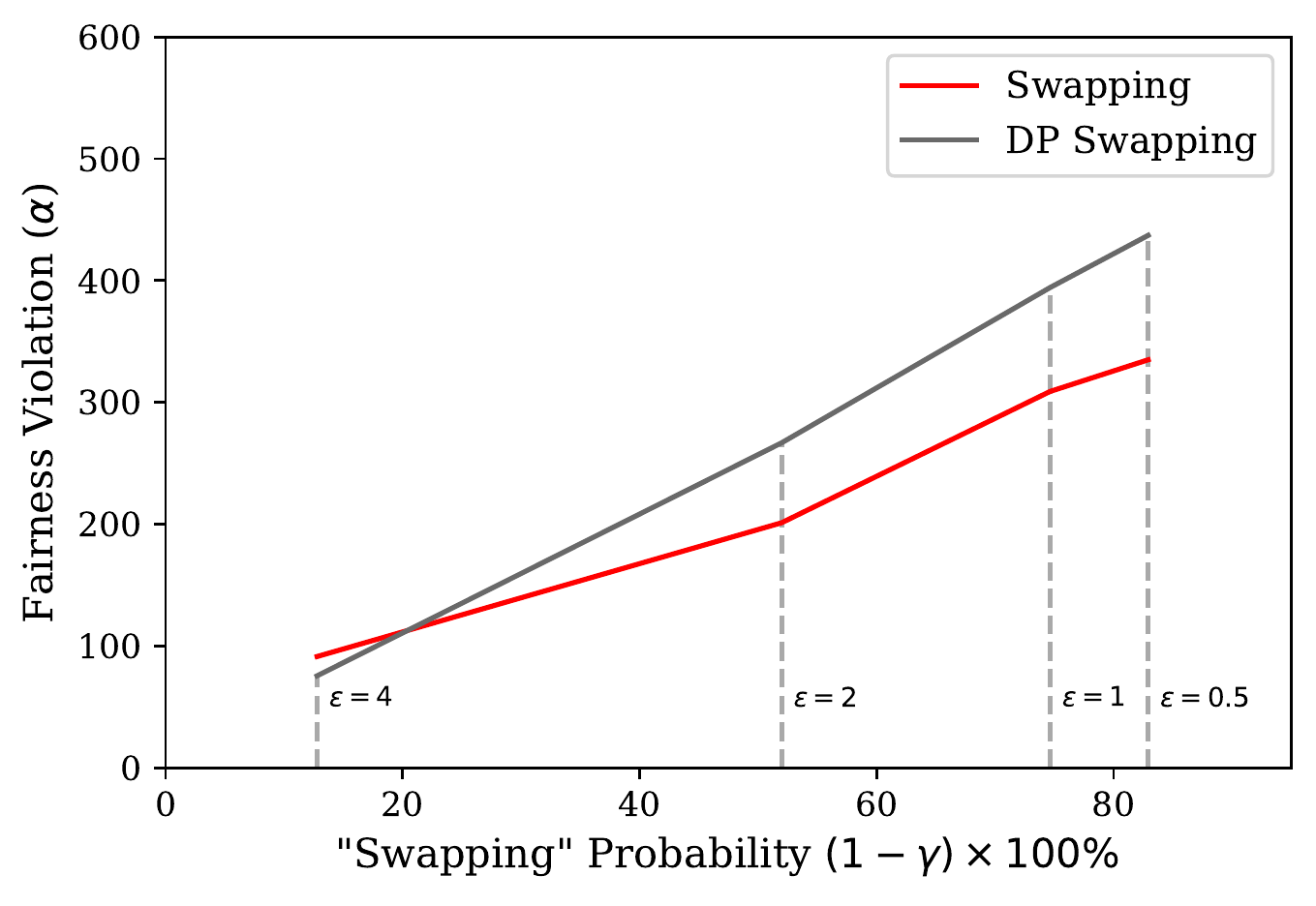}
\includegraphics[width=0.327\textwidth]{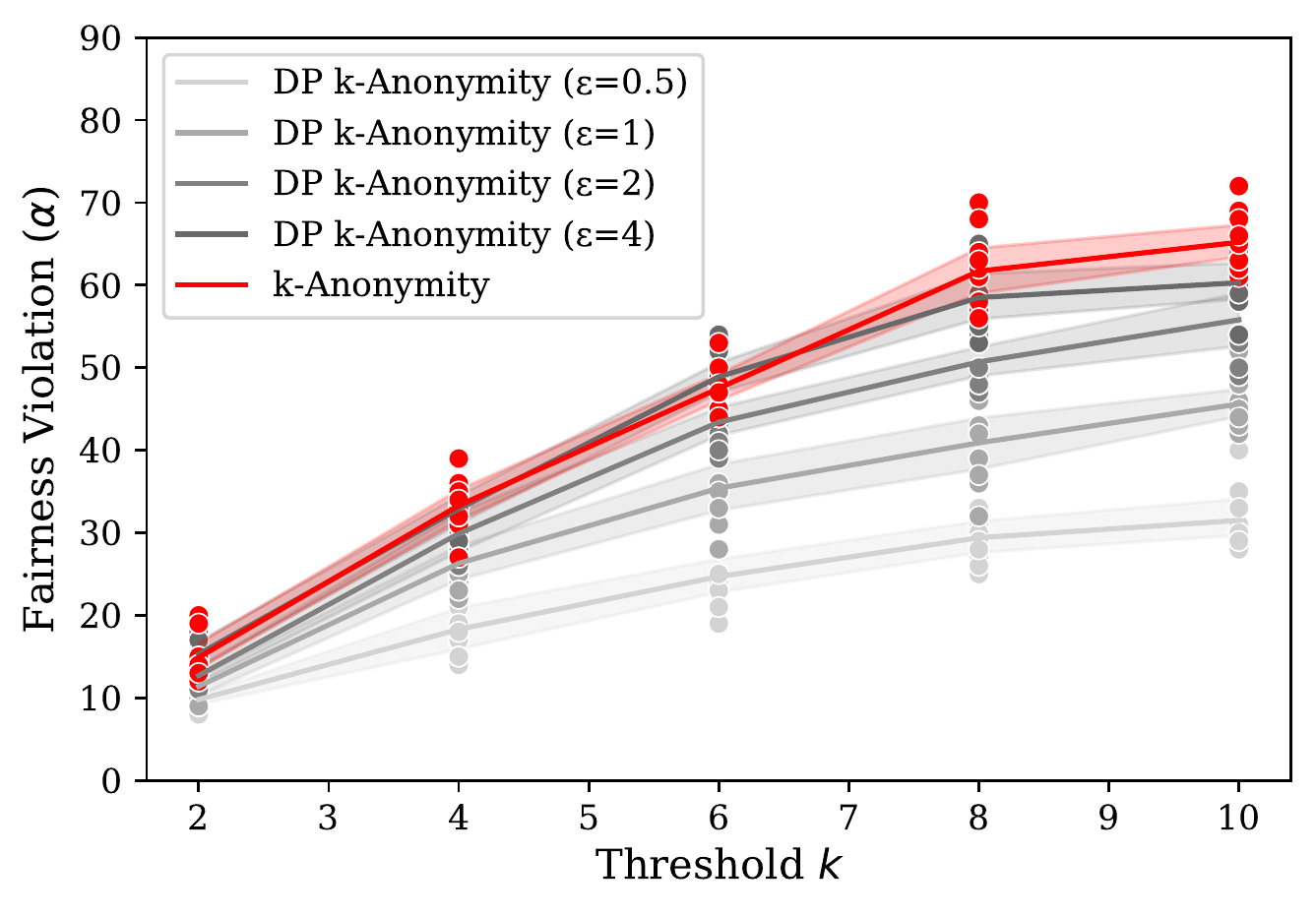}
\caption{TX ACS dataset: Fairness values $\alpha$ for cell suppression (left), 
swapping (center) and $k$-anonymity (right) and their differentially 
private counterparts (average of 200 repetitions).}
\label{fig:TX_fairness}
\end{figure*}

\begin{figure*}[!t]
\centering
\includegraphics[width=0.3\textwidth]{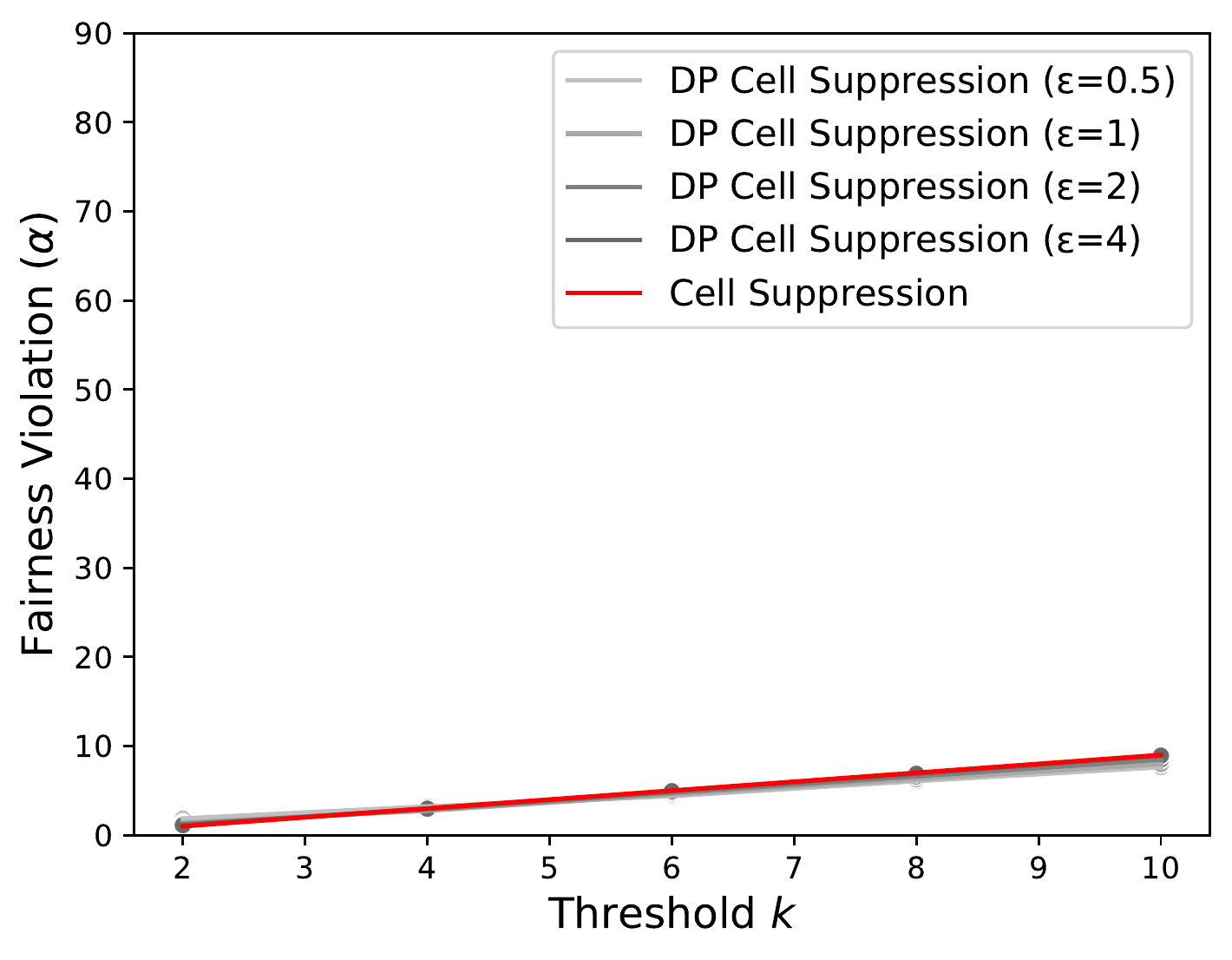}
\includegraphics[width=0.332 \textwidth]{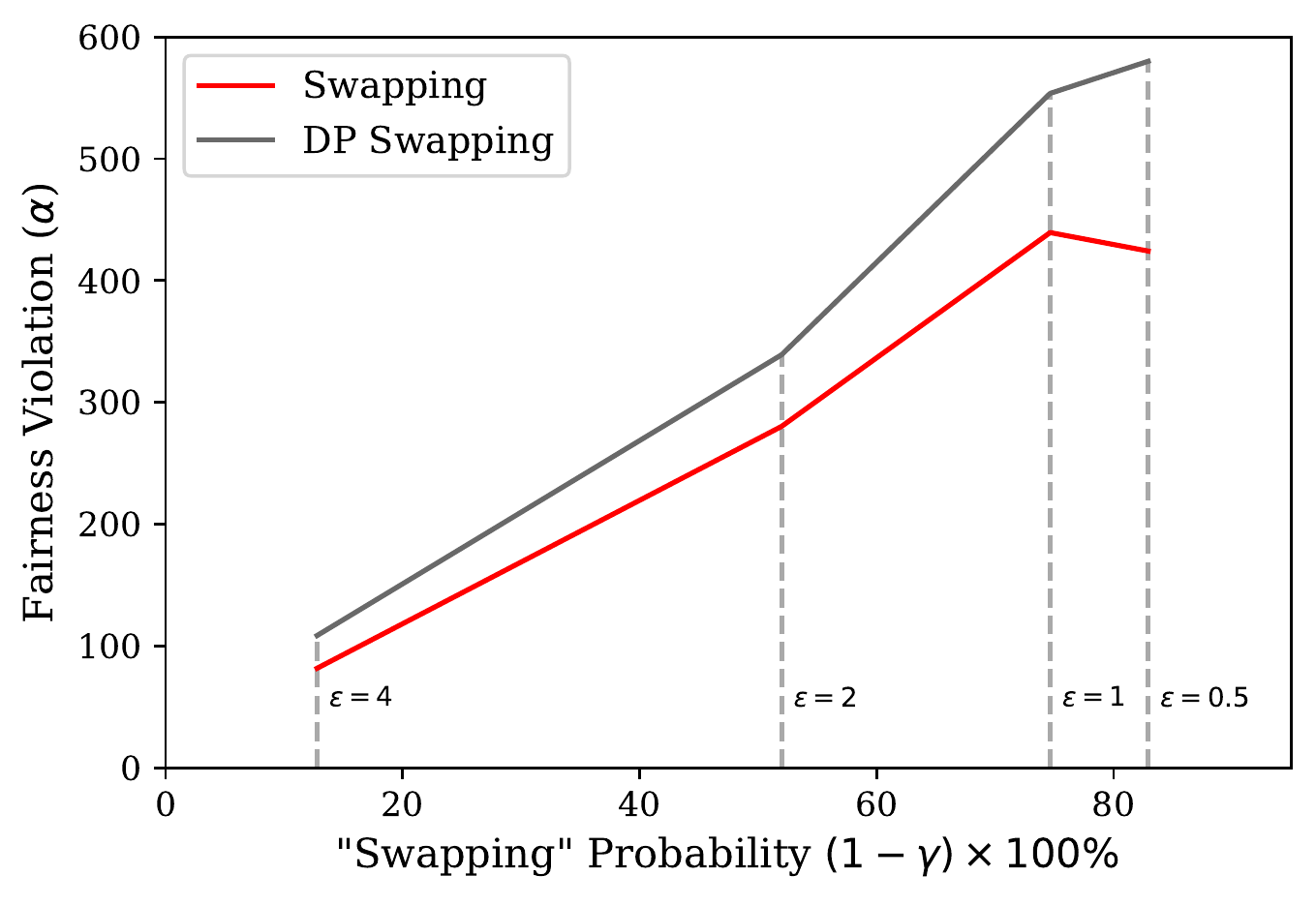}
\includegraphics[width=0.327\textwidth]{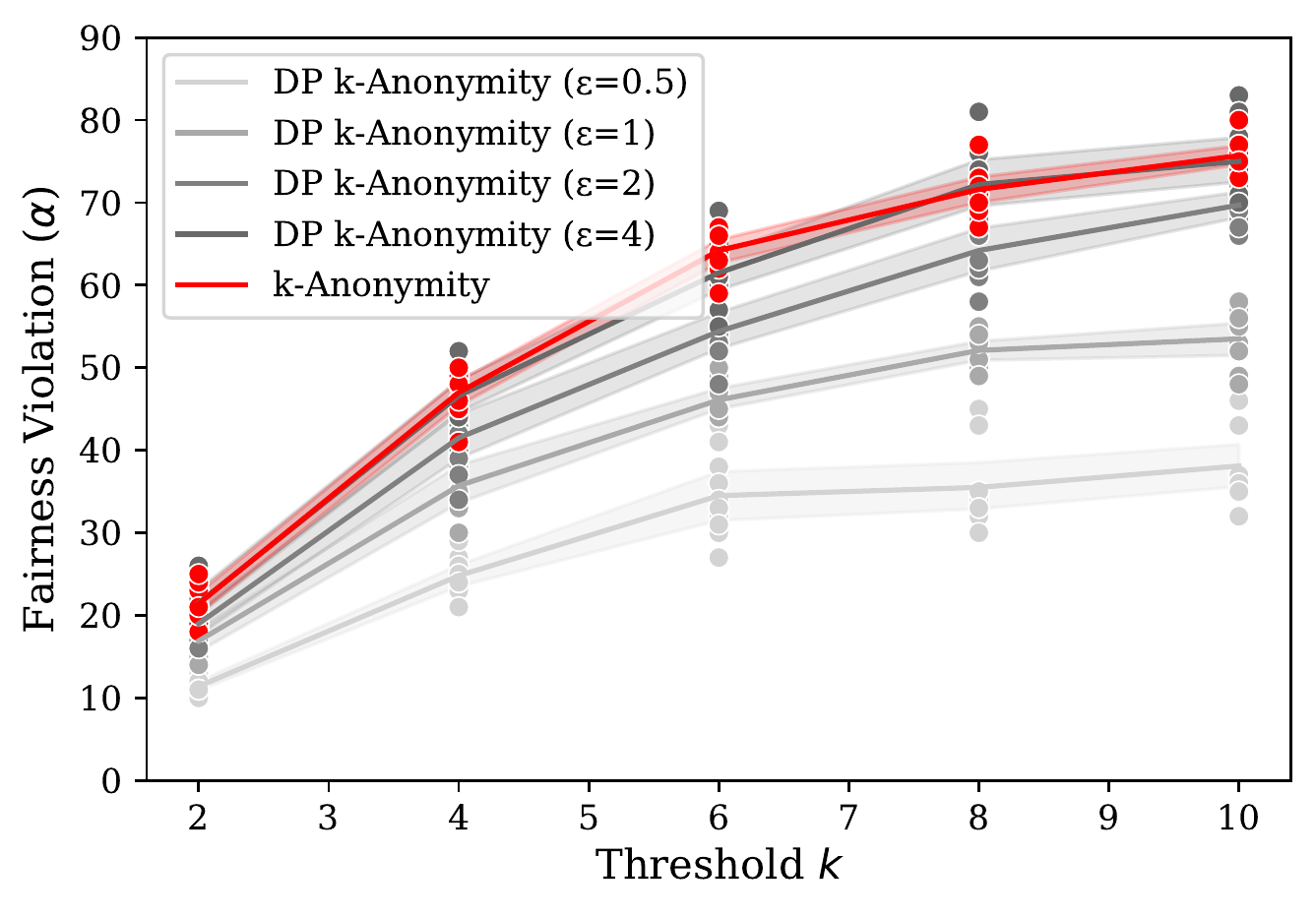}
\caption{Outlier ACS dataset: Fairness values $\alpha$ for cell suppression (left), 
swapping (center) and $k$-anonymity (right) and their differentially 
private counterparts (average of 200 repetitions).}
\label{fig:OT_fairness}
\end{figure*}

\end{document}